\documentclass[final,3p,times]{elsarticle}
\usepackage{amsmath}
\usepackage{amssymb}
\usepackage{wasysym}
\usepackage[nodots]{numcompress}
\usepackage{mathrsfs}
\usepackage[ruled,vlined,noend]{algorithm2e}
\usepackage{mathptmx}
\usepackage{graphicx,color}
\usepackage{epsfig}
\journal{arXiv}
\begin{document}
\begin{frontmatter}
\title{2-connecting Outerplanar Graphs without Blowing up the Pathwidth}
\author[*]{Jasine Babu}
\ead{jasine@csa.iisc.ernet.in}
\author[**]{Manu Basavaraju}
\ead{iammanu@gmail.com}
\author[*]{L. Sunil Chandran}
\ead{sunil@csa.iisc.ernet.in}
\author[*]{Deepak Rajendraprasad}
\ead{deepakr@csa.iisc.ernet.in}
\address[*]{Department of Computer Science and Automation, Indian Institute of Science, Bangalore, India.}
\address[**]{Department of Informatics, University of Bergen, Norway.}
\begin{abstract}
Given a connected outerplanar graph $G$ of pathwidth $p$, we give an algorithm to add edges to $G$ to get a supergraph of $G$, which is $2$-vertex-connected, outerplanar and of pathwidth $O(p)$. This settles an open problem raised by Biedl \cite{Biedl2012}, in the context of computing minimum height planar straight line drawings of outerplanar graphs, with their vertices placed on a two dimensional grid. In conjunction with the result of this paper, the constant factor approximation algorithm for this problem obtained by Biedl \cite{Biedl2012} for $2$-vertex-connected outerplanar graphs will work for all outer planar graphs. 
\end{abstract}
\begin{keyword}
 Pathwidth\sep Outerplanar Graph\sep $2$-vertex-connected
\end{keyword}
\end{frontmatter}
\section{Introduction}
A graph $G(V, E)$ is outerplanar, if it has a planar embedding with all its vertices lying on the outer face. Computing planar straight line drawings of planar graphs, with their vertices placed on a two dimensional grid, is a well known problem in graph drawing. The height of a grid is defined as the smaller of the two dimensions of the grid. If $G$ has a planar straight line drawing, with its vertices placed on a two dimensional grid of height $h$, then we call it a planar drawing of $G$ of height $h$. It is known that any planar graph on $n$ vertices can be drawn on an $(n-1) \times (n-1)$ sized grid \cite{Schnyder1990}. A well studied optimization problem in this context is to minimize the height of the planar drawing. 

Pathwidth is a structural parameter of graphs, which is widely used in graph drawing and layout problems \cite{Biedl2012,Dujmovic2002,Suderman04}. We use $\operatorname{pw}(G)$ to denote the pathwidth of a graph $G$. The study of pathwidth, in the context of graph drawings, was initiated by Dujmovic et al. \cite{Dujmovic2002}. It is known that any planar graph that has a planar drawing of height $h$ has pathwidth at most $h$ \cite{Suderman04}. However, there exist planar graphs of constant pathwidth but requiring $\Omega(n)$ height in any planar drawing \cite{Biedl2011}. In the special case of trees, Suderman \cite{Suderman04} showed that any tree $T$ has a planar drawing of height at most $3 \operatorname{pw}(T)-1$. Biedl \cite{Biedl2012} considered the same problem for the bigger class of outerplanar graphs. For any $2$-vertex-connected outerplanar graph $G$, Biedl \cite{Biedl2012} obtained an algorithm to compute a planar drawing of $G$ of height at most $4\operatorname{pw}(G)-3$. Since it is known that 
pathwidth 
is a lower bound for the height of the 
drawing \cite{Suderman04}, the algorithm given by 
Biedl \cite{Biedl2012} is a $4$-factor approximation algorithm for the problem, for any $2$-vertex-connected outerplanar graph. The method in Biedl \cite{Biedl2012} is to add edges to the $2$-vertex-connected outerplanar graph $G$ to make it a maximal outerplanar graph $H$ and then draw $H$ on a grid of height $4\operatorname{pw}(G)-3$. The same method would give a constant factor approximation algorithm for arbitrary outerplanar graphs, if it is possible to add edges to an arbitrary connected outerplanar graph $G$ to obtain a $2$-vertex-connected outerplanar graph $G'$ such that $\operatorname{pw}(G') = O(\operatorname{pw}(G))$. This was an open problem in Biedl \cite{Biedl2012}.

In this paper, we settle this problem by giving an algorithm to augment a connected outerplanar graph $G$ of pathwidth $p$ by adding edges so that the resultant graph is a $2$-vertex-connected outerplanar graph of pathwidth $O(p)$. Notice that, the non-triviality lies in the fact that $G'$ has to be maintained outerplanar. (If we relax this condition, the problem becomes very easy. It is easy to verify that the supergraph $G'$ of $G$, obtained by making two arbitrarily chosen vertices of $G$ adjacent to each other and to every other vertex in the graph, is $2$-vertex-connected and has pathwidth at most $\operatorname{pw}(G)+2$.) Similar problems of augmenting outerplanar graphs to make them $2$-vertex-connected, while maintaining the outerplanarity and optimizing some other properties, like number of edges added \cite{Garcia10,Kant96}, have also been investigated previously. 
\section{Background}\label{background}
A \textit{tree decomposition} of a graph $G(V, E)$ \cite{RobertsonS84} is a pair $(T, \mathscr{X})$, where $T$ is a tree
and $\mathscr{X} = (X_t : t \in V(T))$ is a family of subsets of $V(G)$, such that: 
\begin{enumerate}
 \item $\bigcup(X_t : t \in V(T)) = V(G)$.
 \item For every edge $e$ of $G$ there exists $t \in V(T)$ such that $e$ has both its end points in $X_t$.
 \item For every vertex $v \in V$, the induced subgraph of $T$ on the vertex set $\{t \in V(T): v \in X_t \}$ is connected. 
\end{enumerate}
The width of the tree decomposition is $\max_{t\in V(T)}{(|X_t| - 1)}$. Each $X_t \in \mathscr{X}$ is referred to as a bag in the tree decomposition. A graph $G$ has \textit{treewidth} $w$ if $w$ is the minimum integer such that $G$ has a tree decomposition of width $w$.

A \textit{path decomposition} $(P, \mathscr{X})$ of a graph $G$ is a tree decomposition of $G$ with the additional property that the tree $P$ is a path. The width of the path  decomposition is $\max_{t\in V(P)}{(|X_t| - 1)}$. A graph $G$ has \textit{pathwidth} $w$ if $w$ is the minimum integer such that $G$ has a path decomposition of width $w$. 

Without loss of generality we can assume that, in any path decomposition $(\mathcal{P}, \mathscr{X})$ of $G$, the vertices of the path $\mathcal{P}$ are labeled as $1, 2,\ldots$, in the order in which they appear in $\mathcal{P}$. Accordingly, the bags in $\mathscr{X}$ also get indexed as $1, 2,\ldots$. For each vertex $v \in V(G)$, define $FirstIndex_\mathscr{X}(v)=\min\{i \mid X_i \in \mathscr{X}$ contains $v\}$,  $LastIndex_\mathscr{X}(v)=\max\{i \mid X_i \in \mathscr{X}$ contains $v\}$ and $Range_\mathscr{X}(v)=[FirstIndex_\mathscr{X}(v), LastIndex_\mathscr{X}(v)]$. By the definition of a path decomposition, if $t \in Range_\mathscr{X}(v)$, then $v \in X_t$. If $v_1$ and $v_2$ are two distinct vertices, define $Gap_\mathscr{X}(v_1,v_2)$ as follows: 
\begin{itemize}                                                                                                                                                                                                                                                                                                                                                                                                                                                                                                                        
\item If $Range_\mathscr{X}(v_1) \cap Range_\mathscr{X}(v_2) \ne \emptyset$, then $Gap_\mathscr{X}(v_1,v_2)=\emptyset$.
\item If $LastIndex_\mathscr{X}(v_1)$ $<$ $FirstIndex_\mathscr{X}(v_2)$, then\\$Gap_\mathscr{X}(v_1,v_2)=[LastIndex_\mathscr{X}(v_1)+1, FirstIndex_\mathscr{X}(v_2)]$.
\item If $LastIndex_\mathscr{X}(v_2)$ $<$ $FirstIndex_\mathscr{X}(v_1)$, then\\$Gap_\mathscr{X}(v_1,v_2)=[LastIndex_\mathscr{X}(v_2)+1, FirstIndex_\mathscr{X}(v_1)]$.
\end{itemize}

The motivation for this definition is the following. Suppose $(\mathcal{P}, \mathscr{X})$ is a path decomposition of a graph $G$ and $v_1$ and $v_2$ are two non-adjacent vertices of $G$. If we add a new edge between $v_1$ and $v_2$, a natural way to modify the path decomposition to reflect this edge addition is the following. If $Gap_\mathscr{X}(v_1,v_2)=\emptyset$, there is already an $X_t \in \mathscr{X}$, which contains $v_1$ and $v_2$ together and hence, we need not modify the path decomposition. If $LastIndex_\mathscr{X}(v_1)$ $<$ $FirstIndex_\mathscr{X}(v_2)$, we insert $v_1$ into all $X_t \in \mathscr{X}$, such that $t \in Gap_\mathscr{X}(v_1,v_2)$. On the other hand, if $LastIndex_\mathscr{X}(v_2)$ $<$ $FirstIndex_\mathscr{X}(v_1)$, we insert $v_2$ to all $X_t \in \mathscr{X}$, such that $t \in Gap_\mathscr{X}(v_1,v_2)$. It is clear from the definition of $Gap_\mathscr{X}(v_1,v_2)$ that this procedure gives a path decomposition of the new graph. Whenever we add an edge $(v_1, v_2)$, we stick to this 
procedure to update 
the path decomposition.

A \textit{block} of a connected graph $G$ is a maximal connected subgraph of $G$ without a cut vertex. Every block of a connected graph $G$ is thus either a single edge which is a bridge in $G$, or a maximal $2$-vertex-connected subgraph of $G$. If a block of $G$ is not a single edge, we call it a non-trivial block of $G$. Given a connected outerplanar graph $G$, we define a rooted tree $T$ (hereafter referred to as the \textit{rooted block tree} of $G$) as follows. The vertices of $T$ are the blocks of $G$ and the root of $T$ is an arbitrary block of $G$ which contains at least one non-cut vertex (it is easy to see that such a block always exists). Two vertices $B_i$ and $B_j$ of $T$ are adjacent if the blocks $B_i$ and $B_j$ share a cut vertex in $G$. It is easy to see that $T$, as defined above, is a tree. In our discussions, we restrict ourselves to a fixed rooted block tree of $G$ and all the definitions hereafter will be with respect to this chosen tree. If block $B_i$ is a child block of block $B_j$ 
in the rooted block tree of $G$, and they share a cut vertex $x$, we say that $B_i$ is a 
child block of $B_j$ at $x$. 

It is known that every $2$-vertex-connected outerplanar graph has a unique Hamiltonian cycle \cite{Syslo79}. Though the Hamiltonian cycle of a $2$-vertex-connected block of $G$ can be traversed either clockwise or anticlockwise, let us fix one of these orderings, so that the \textbf{successor} and \textbf{predecessor} of each vertex in the Hamiltonian cycle in a block is fixed. We call this order the clockwise order. Consider a non-root block $B_i$ of $G$ such that $B_i$ is a child block of its parent, at the cut vertex $x$. If $B_i$ is a non-trivial block and $y_i$ and $y'_i$ respectively are the predecessor and successor of $x$ in the Hamiltonian cycle of $B_i$, then we call $y_i$ the last vertex of $B_i$ and $y'_i$ the first vertex of $B_i$. If $B_i$ is a trivial block, the sole neighbor of $x$ in $B_i$ is regarded as both the first vertex and the last vertex of $B_i$. By the term  \textbf{path}, we always mean a simple path, i.e., a path in which no vertex repeats. 
\section{An overview of our method}
Given a connected outerplanar graph $G(V, E)$ of pathwidth $p$, our algorithm will produce a $2$-vertex-connected outerplanar graph $G''(V, E'')$ of pathwidth $O(p)$, where $E \subseteq E''$. Our algorithm proceeds in three stages. 

(1) We use a modified version of the algorithm proposed by Govindan et al. \cite{Govindan98} to obtain a \textit{nice tree decomposition} (defined in Section \ref{stage1}) of $G$. Using this nice tree decomposition of $G$, we construct a special path decomposition of $G$ of width at most $4p+3$.

(2) For each cut vertex $x$ of $G$, we define an ordering among the child blocks attached through $x$ to their parent block. To define this ordering, we use the special path decomposition constructed in the first stage. This ordering helps us in constructing an outerplanar supergraph $G'(V, E')$ of $G$, whose pathwidth is at most $8p+7$, such that for every cut vertex $x$ in $G'$, $G' \setminus x$ has exactly two components. The properties of the special path decomposition of $G$ obtained in the first stage is crucially used in our argument to bound the width of the path decomposition of $G'$, produced in the second stage. 

(3) We $2$-vertex-connect $G'$ to construct $G''(V, E'')$, using a straightforward algorithm. As a by-product, this algorithm also gives us a surjective mapping from the cut vertices of $G'$ to the edges in $E'' \setminus E'$. We give a counting argument based on this mapping and some basic properties of path decompositions to show that the width of the path decomposition of $G''$ produced in the third stage is at most $16p+15$. 
\section{Stage 1: Construct a nice path decomposition of $G$}\label{stage1}
In this section, we construct a \textit{nice tree decomposition} of the connected outerplanar graph $G$ and then use it to construct a \textit{nice path decomposition} of $G$. We begin by giving the definition of a nice tree decomposition. 

Given an outerplanar graph $G$, Govindan et al. \cite[Section 2]{Govindan98} gave a linear time algorithm to construct a width $2$ tree decomposition $(T, \mathscr{Y})$ of $G$ where $\mathscr{Y}=(Y_t : t \in V(T))$, with the following special properties:
\begin{enumerate}
 \item  There is a bijective mapping $\mathit{b}$ from $V(G)$ to $V(T)$ such that, for each $v\in V(G)$, $v$ is present in the bag $Y_{\mathit{b}(v)}$.
 \item  If $B_i$ is a child block of $B_j$ at a cut vertex $x$, the vertex set $\left\{\mathit{b}(v)\mid v\in V(B_i\setminus x) \right\}$ induces
a subtree $T'$ of $T$ such that, if $\left(\mathit{b}(u), \mathit{b}(v)\right)$ is an edge in $T'$, then $(u, v) \in E(G)$ - this means that the subgraph $B_i\setminus x$ of $G$ has a spanning tree, which is a copy of $T'$ on the corresponding vertices. Moreover, 
$(T', \mathscr{Y}')$, with $\mathscr{Y}'=(Y_t : t \in V(T'))$ gives a tree decomposition of $B_i$.
\item $G$ has a spanning tree, which is a copy of $T$ on the corresponding vertices; i.e. if $\left(\mathit{b}(u), \mathit{b}(v)\right)$ is an edge in $T$, then $(u, v) \in E(G)$. 
 \end{enumerate}
\newdefinition{definition}{Definition}
\begin{definition}[Nice tree decomposition of an outerplanar graph $G$]
 A tree decomposition $(T, \mathscr{Y})$ of $G$, where $\mathscr{Y}=(Y_t : t \in V(T))$ having properties 1, 2 and 3 above, together with the following additional property, is called a nice tree decomposition of $G$.
\begin{enumerate}
\item[4.] If $y_i$ and $y'_i$ are respectively the last and first vertices of a non-root, non-trivial block $B_i$, then the bag $Y_ {\mathit{b}(y_i)} \in \mathscr{Y}$ contains both $y_i$ and $y'_i$. 
\end{enumerate}
\end{definition}
In the discussion that follows, we will show that any outerplanar graph $G$ has a nice tree decomposition $(T, \mathscr{Y})$ of width at most $3$. Initialize $(T, \mathscr{Y})$ to be the tree decomposition of $G$, constructed using the method proposed by Govindan et al. \cite{Govindan98}, satisfying properties 1, 2 and 3 of nice tree decompositions. We need to modify $(T, \mathscr{Y})$ in such a way that, it satisfies property 4 as well. 

For every non-root, non-trivial block $B_i$ of $G$, do the following. If $y_i$ and $y'_i$ are respectively the last and first vertices of $B_i$, then, for each $t \in \left\{\mathit{b}(v)\mid v\in V(B_i\setminus x) \right\}$, we insert $y'_i$ to $Y_t$, if it is not already present in $Y_t$ and we call $y'_i$ as a \textit{propagated} vertex. Note that, after this modification $Y_{\mathit{b}(y_i)}$ contains both $y_i$ and $y'_i$. 
\newtheorem{claim}{Claim}
\begin{claim} \label{claim1}
 After the modification, $(T, \mathscr{Y})$ remains a tree decomposition of $G$.
\end{claim}
\newproof{proof}{Proof}
\begin{proof}
Clearly, we only need to verify that the third property in the definition of a tree decomposition holds, for all the propagated vertices. Let $y'_i$ be a propagated vertex, which got inserted to the bags corresponding to vertices of $B_i\setminus x$, during the modification. Let $V_{y'_i}=\{t \mid y'_i \in Y_t$, before the modification$\}$ and let $V'_{y'_i}=\{t \mid y'_i \in Y_t$, after the modification$\}$. Then, clearly, $V'_{y'_i} = V_{y'_i}\cup \left\{\mathit{b}(v)\mid v\in V(B_i\setminus x) \right\}$. 

Clearly, the induced subgraph of $T$ on the vertex set $V_{y'_i}$ is connected, since we had a tree decomposition of $G$ before the modification. By property 2 of nice decompositions, the induced subgraph of $T$ on the vertex set $\left\{\mathit{b}(v)\mid v\in V(B_i\setminus x) \right\}$ is also connected. Moreover, by property 1 of nice decompositions, $\mathit{b}(y'_i) \in V_{y'_i}$ and hence, $\mathit{b}(y'_i) \in \left\{\mathit{b}(v)\mid v\in V(B_i\setminus x) \right\} \cap V_{y'_i}$. This implies that the induced subgraph of $T$ on the vertex set $V'_{y'_i}$ is connected. 
\qed
\end{proof}
\begin{claim} \label{claim2}
 After the modification, $(T, \mathscr{Y})$ becomes a nice tree decomposition of $G$ of width at most $3$.
\end{claim}
\begin{proof}
It is easy to verify that all the four properties required by nice decompositions are satisfied, after the modification. Moreover, for any block $B_i$, attached to its parent at the cut vertex $x$, at most one (propagated) vertex is getting newly inserted into the bags corresponding to vertices of $B_i \setminus x$. Since the size of any bag in $\mathscr{Y}$ was at most two initially and it got increased by at most one, the new decomposition has width at most three. 
\qed
\end{proof}
From the claims above, we can conclude the following.
\newtheorem{lemma}{Lemma}
\begin{lemma}\label{nicetree}
Every outerplanar graph $G$ has a nice tree decomposition $(T, \mathscr{Y})$ of width $3$, constructible in polynomial time.
\end{lemma}
\begin{definition}[Nice path decomposition of an outerplanar graph]
 Let $(\mathcal{P}, \mathscr{X})$ be a path decomposition of an outerplanar graph $G$. If, for every non-root non-trivial block $B_i$, there is a bag $X_t \in \mathscr{X}$ containing both the first and last vertices of $B_i$ together, then $(\mathcal{P}, \mathscr{X})$ is called a nice path decomposition of $G$.
\end{definition}
\begin{lemma}\label{nicepath}
Let $G$ be an outerplanar graph with $\operatorname{pw}(G)=p$. A nice path decomposition $(\mathcal{P}, \mathscr{X})$ of $G$, of width at most $4p+3$, is constructible in polynomial time. 
\end{lemma}
\begin{proof}
Let $(T, \mathscr{Y})$, with $\mathscr{Y}=(Y_{v_{_T}} : v_{_T} \in V(T))$ be a nice tree decomposition of $G$ of width $3$, obtained using Lemma \ref{nicetree}. Obtain an optimal path decomposition $(\mathcal{P}_T, \mathscr{X}_T)$ of the tree $T$ in polynomial time, using a standard algorithm (For example, the algorithm from \cite{Skodinis2003}). Since $T$ is a spanning tree of $G$, the pathwidth of $T$ is at most that of $G$. Therefore, the width of the path decomposition $(\mathcal{P}_T, \mathscr{X}_T)$ is at most $p$; i.e. there are at most $p+1$ vertices of $T$ in each bag $X_{T_i} \in \mathscr{X}_T$.

Let $\mathcal{P}= \mathcal{P}_T$ and for each $X_{T_i} \in \mathscr{X}_T$, we define $X_{i} = \bigcup_{v_{_T} \in X_{T_i}}{Y_{v_{_T}}}$. It is not difficult to show that $(\mathcal{P}, \mathscr{X})$, with $\mathscr{X}=(X_1,\ldots, X_{|V(\mathcal{P}_T)|})$, is a path decomposition of $G$ (See \cite{Govindan98}). The width of this path decomposition is at most $4(p+1)-1=4p+3$, since $|Y_{v_{_T}}| \le 4$, for each bag $Y_{v_{_T}} \in \mathscr{Y}$ and $|X_{T_{i}}| \le p+1$, for each bag $X_{T_{i}} \in \mathscr{X}_T$. 

Let $B_i$ be a non-root, non-trivial block  in $G$ and $y_i$ and $y'_i$ respectively be the first and last vertices of $B_i$. Since $\mathit{b}(y_i)$ is a vertex of the tree $T$, there is some bag $X_{T_j} \in \mathscr{X}_T$, containing $\mathit{b}(y_i)$. The bag $Y_ {\mathit{b}(y_i)}\in \mathscr{Y}$ contains both $y_i$ and $y'_i$, since $(T, \mathscr{Y})$ is a nice tree decomposition of $G$. It follows from the definition of $X_j$ that $X_j \in \mathscr{X}$ contains both $y_i$ and $y'_i$. Therefore, $(\mathcal{P}, \mathscr{X})$ is a nice path decomposition of $G$.
\qed
\end{proof}
\section{Edge addition without spoiling the outerplanarity}
In this section, we prove some technical lemmas which will be later used to prove that the intermediate graph $G'$ obtained in Stage 2 and the $2$-vertex-connected graph $G''$ obtained in Stage 3 are outerplanar. To get an intuitive understanding of these lemmas, the reader may refer to Fig \ref{fig1}. Recall that, when we use the term \textit{path}, it always refers to a simple path. \vspace{-0.15cm}
\begin{lemma}\label{extensionlemma}
 Let $G(V, E)$ be a connected outerplanar graph. Let $u$ and $v$ be two distinct non-adjacent vertices in $G$ and let $P= (u=x_0, x_1, x_2, \ldots, x_k, x_{k+1}=v)$ where $k\ge 1$ be a path in $G$ such that:\vspace{-0.15cm}
\begin{description}
 \item[(i)] $P$ shares at most one edge with any block of $G$. \vspace{-0.15cm}
 \item[(ii)] For $0 \le i \le k$, if the block containing the edge $(x_i, x_{i+1})$ is non-trivial, then $x_{i+1}$ is the successor of $x_i$ in the Hamiltonian cycle of that block.
\end{description}\vspace{-0.15cm}
 Then the graph $G'(V, E')$, where $E'=E \cup \{(u, v)\}$, is outerplanar.
\end{lemma}
\clearpage
\begin{figure}
 \begin{center}
\includegraphics[scale=0.72]{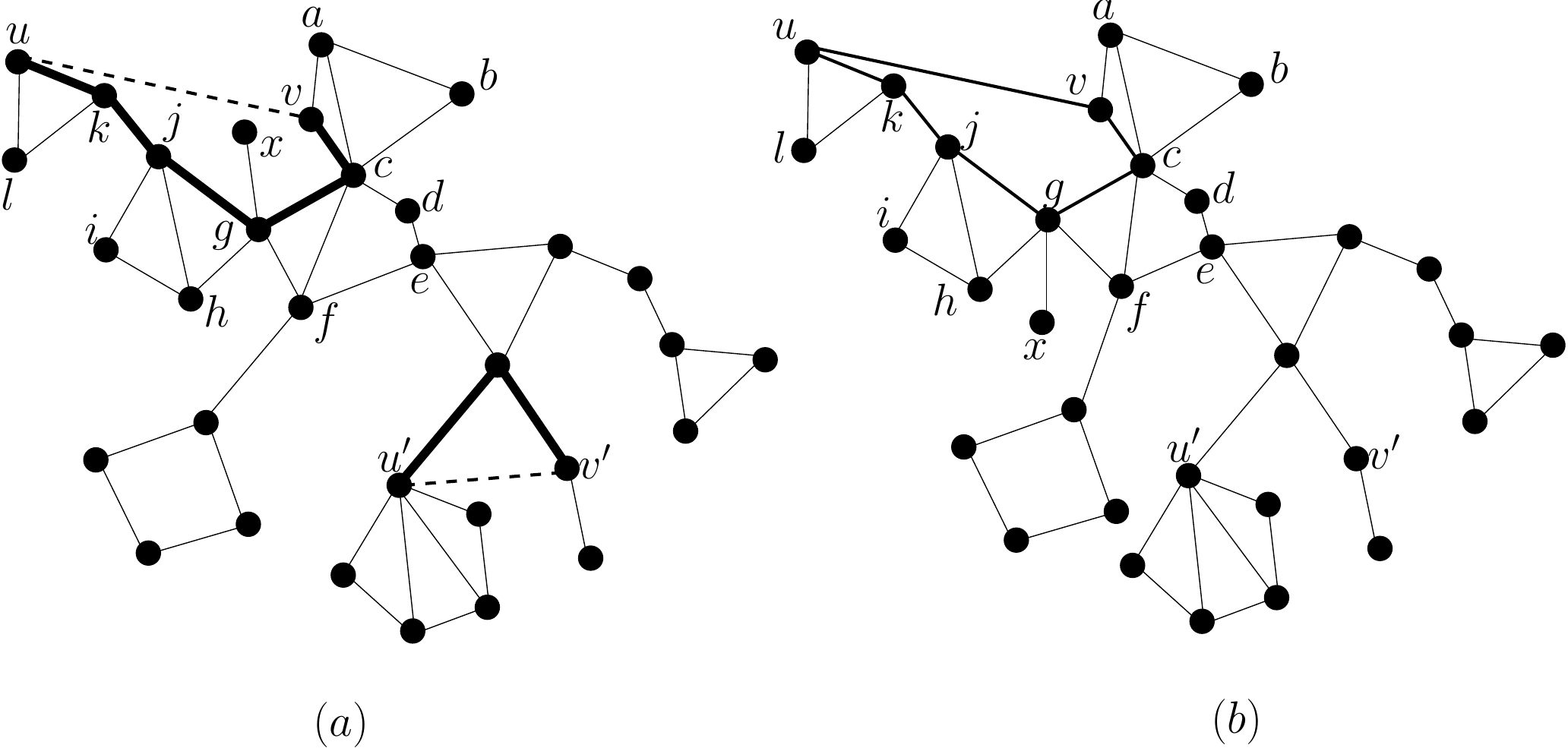}
\caption{(a) The path between $u$ and $v$ and the path between $u'$ and $v'$ (shown in thick edges) satisfy the conditions stated in Lemma \ref{extensionlemma}. According to Lemma \ref{extensionlemma}, on adding any one of the dotted edges $(u, v)$ or $(u', v')$, the resultant graph is outerplanar. (b) An outerplanar drawing of the resultant graph, after adding the edge $(u, v)$. In this graph, $u, v, a, b, c, d, e, f, g, h, i, j, k, l, u$ is the Hamiltonian cycle of the new block formed.}
\label{fig1}
\end{center}
\end{figure}
\begin{proof}
It is well known that a graph $G$ is outerplanar if and only if it contains no subgraph that is a subdivision of $K_4$ or $K_{2,3}$ \cite{Chartrand67}. Consider a path $P$ between $u$ and $v$ as stated in the lemma. 
\newtheorem{property}{Property}
\begin{property}\label{mustappear}
 In every path in $G$ from $u$ to $v$, vertices $x_1, \ldots, x_k$ should appear and for $0 \le i \le k$, $x_i$ should appear before $x_{i+1}$ in any such path.
\end{property}
\begin{proof}
For any $1 \le i \le k$, the two consecutive edges $e_i=(x_{i-1}, x_i)$ and $e_{i+1}=(x_i, x_{i+1})$ of the path $P$ belong to two different blocks of $G$, by assumption. Therefore, each internal vertex $x_i$, $1 \le i \le k$, is a cut vertex in $G$. As a result, in every path in $G$ between $u$ and $v$, vertices $x_1, \ldots, x_k$ should appear and for $0 \le i \le k$, $x_i$ should appear before $x_{i+1}$ in any such path.
\qed
\end{proof}
\begin{property}\label{propdisjoint}
For any $0 \le i \le k$, there are at most two internally vertex disjoint paths in $G$ between $x_i$ and $x_{i+1}$.
\end{property}
\begin{proof}
Any path from $x_i$ to $x_{i+1}$ lies fully inside the block $B_i$ that contains the edge $(x_i, x_{i+1})$. If $B_i$ is trivial, the only path from $x_i$ to $x_{i+1}$ is the direct edge between them. If this is not the case, $B_i$ is $2$-vertex-connected. It is a property of $2$-vertex-connected outerplanar graphs that, the number of internally vertex disjoint paths between any two consecutive vertices of the Hamiltonian cycle of the graph, is exactly two  (See Appendix). By the assumption of Lemma \ref{extensionlemma}, if $B_i$ is non-trivial, then $x_{i+1}$ is the successor of $x_i$ in the Hamiltonian cycle of $B_i$. Hence, the property follows. 
\qed
\end{proof}
We will show that if $G'$ is not outerplanar, then $G$ also was not outerplanar, which is a contradiction. Assume that $G'$ is not outerplanar. This implies that there is a subgraph $H'$ of $G'$ that is a subdivision of $K_4$ or $K_{2,3}$. Since $G$ does not have a subgraph that is a subdivision of $K_4$ or $K_{2,3}$, $H'$ cannot be a subgraph of $G$. Hence, the new edge $(u, v)$ should be an edge in $H'$ and all other edges of $H'$ are edges of $G$.\newline\newline
Case 1. $H'$ is a subdivision of $K_4$. 

Let $k_1, k_2, k_3$ and $k_4$ denote the four vertices of $H'$ that correspond to the vertices of $K_4$. We call them as branch vertices of $H'$. For $i, j \in \{1, 2, 3, 4\}$, $i \ne j$, let $P_{i, j}$ denote the path in $H'$ from the branch vertex $k_i$ to the branch vertex $k_j$, such that each intermediate vertex of the path is a degree two vertex in $H'$. Without loss of generality, assume that the edge $(u, v)$ is part of the path $P_{1, 2}$ of $H'$. 
\begin{claim}
All of the vertices $x_1, \ldots, x_k$ appear in $P_{1, 2}$. The order $<$ in which the vertices $u$, $v$, $x_1, \ldots, x_k$ appear in $P_{1, 2}$ should be one of the following three: (without loss of generality, assuming $u <v$): (1) $u <v<x_k<x_{k-1}<\cdots< x_1$\\(2) $x_k <x_{k-1} < \cdots< x_1<u<v$ (3) $x_j < x_{j-1} < \cdots< x_1< u < v < x_k \cdots <x_{j+1}$ for some $j \in \{1, 2, \ldots, k-1\}$. 
\end{claim}
\begin{proof}
Suppose $x_i$, $1 \le i \le k$, does not belong to the path $P_{1, 2}$. Then, there is a path in $H'\setminus (u, v)$ between vertices $u$ and $v$, avoiding the vertex $x_i$, since $H'$ is a subdivision of $K_4$. Since $H'\setminus (u, v)$ is a subgraph of $G$, this implies that there is a path in $G$, between $u$ and $v$ that avoids $x_i$. This is a contradiction to Property \ref{mustappear}. Therefore, $x_i \in P_{1, 2}$. Notice that there is a path in $H'\setminus (u, v)$, and hence in $G$, between $u$ and $v$ that goes through the vertex $k_3$. To satisfy Property \ref{mustappear}, $x_i$ should appear before $x_{i+1}$, for $0 \le i \le k$, in this path. Hence, one of the orderings mentioned in the claim should happen in $P_{1, 2}$. 
\qed
\end{proof}
Let us denote the first vertex in the ordering $<$ by $z_1$ and the last vertex in the ordering $<$ by $z_2$. (In the first case, $z_1 = u$ and $z_2 =x_1$. In the second case, $z_1=x_k$ and $z_2=v$. In the third case, $z_1=x_j$ and $z_2 =x_{j+1}$.) In all the three cases of the ordering, there is a direct edge in $G$, between $z_1$ and $z_2$. Notice that in any of these three possible orderings, we do not have $z_1=u$ and $z_2=v$ simultaneously. Since $(z_1, z_2) \ne (u, v)$, by deleting the intermediate vertices between $z_1$ and $z_2$ from the path $P_{1, 2}$ and including the direct edge between $z_1$ and $z_2$, we get a path $P'_{1, 2}$ between $k_1$ and $k_2$ in $G$. All vertices in $P'_{1, 2}$ are from the vertex set of $P_{1, 2}$. Therefore, by replacing the path $P_{1, 2}$ in $H'$ by $P'_{1, 2}$, we get a subgraph $H$ of $G$ that is a subdivision of $K_4$. This means that $G$ is not outerplanar, which is a contradiction. Therefore, $H'$ cannot be a subdivision of $K_4$.\newline\newline 
Case 2. $H'$ is a subdivision of $K_{2,3}$. 

As earlier, let $k_1, k_2, k_3, k_4$ and $k_5$ denote the branch vertices of $H'$ that correspond to the vertices of $K_{2, 3}$. Let $k_1, k_3, k_5$ be the degree $2$ branch vertices in $H'$ and $k_2, k_4$ be the degree $3$ branch vertices of $H'$. For $i \in \{1, 3, 5\}$ and $j \in \{2, 4\}$, let $P_{i, j}$ denote the path in $H'$ from vertex $k_i$ to vertex $k_j$, such that each intermediate vertex of the path is a degree two vertex in $H'$. Also, for $i \in \{1, 3, 5\}$ and $j \in \{2, 4\}$ let $P_{j, i}$ denote the path from $j$ to $i$ in which the vertices in $P_{j, i}$ appear in the reverse order compared to $P_{i, j}$. Without loss of generality, assume that the edge $(u, v)$ is part of the path $P_{1, 2}$ of $H'$. Let $P_{4, 1, 2}$ denote the path in $H'$ between vertices $k_4$ and $k_2$, obtained by concatenating the paths $P_{4, 1}$ and $P_{1, 2}$. 
\begin{claim}
All of the vertices $x_1, \ldots, x_k$ appear in $P_{4, 1, 2}$. The order $<$ in which the vertices $u$, $v$, $x_1, \ldots, x_k$ appear in $P_{4, 1, 2}$ should be one of the following three (without loss of generality, assuming $u<v$): (1) $u <v<x_k<x_{k-1}<\cdots< x_1$\\(2) $x_k <x_{k-1} < \cdots< x_1<u<v$ (3) $x_j < x_{j-1} < \cdots< x_1< u < v < x_k \cdots <x_{j+1}$ for some $j \in \{1, 2, \ldots, k-1\}$. 
\end{claim}
This can be proved in a similar way as in Case 1. The remaining part of the proof is also similar. Let us denote the first vertex in the ordering $<$ by $z_1$ and the last vertex in the ordering $<$ by $z_2$. Repeating similar arguments as in Case 1, we can prove that by deleting the intermediate vertices between $z_1$ and $z_2$ from the path $P_{4, 1, 2}$ and including the direct edge between $z_1$ and $z_2$, we get a path $P'_{4, 1, 2}$ between $k_4$ and $k_2$ in $G$. All vertices in $P'_{4, 1, 2}$ are from the vertex set of $P_{4, 1, 2}$. Therefore, by replacing the path $P_{4, 1, 2}$ in $H'$ by $P'_{4, 1, 2}$, we get a subgraph $H$ of $G$. If $P'_{4, 1, 2}$ has at least one intermediate vertex, the subgraph $H$ of $G$, obtained by replacing the path $P_{4, 1, 2}$ in $H'$ by $P'_{4, 1, 2}$, is a subdivision of $K_{2, 3}$, where an intermediate vertex of $P'_{4, 1, 2}$ takes the role of the branch vertex $k_1$. This contradicts the assumption that $G$ is outerplanar. 

Therefore, assume that $P'_{4, 1, 2}$ has no intermediate vertices, i. e., $z_1=k_4$ and $z_2=k_2$. In the first case of ordering $<$ mentioned in the claim above, we have $k_4=z_1=u=x_0$ and $k_2=z_2=x_1$. In the second case, $k_4=z_1=x_k$ and $k_2=z_2=v=x_{k+1}$. In the third case, $k_4=z_1=x_j$ and $k_2=z_2=x_{j+1}$. In each of these cases, by Property \ref{propdisjoint}, there can be at most two vertex disjoint paths in $G$ between $z_1$ and $z_2$. But, in all these cases, there is a direct edge between $k_4=z_1$ and $k_2=z_2$ in $G$. Since $H'$ is a subdivision of $K_{2, 3}$, other than this direct edge, in $H'\setminus (u, v)$ there are two other paths from $k_4=z_1$ to $k_2=z_2$ that are internally vertex disjoint and containing at least one intermediate vertex. This will mean that there are at least three internally vertex disjoint paths from $z_1=k_4$ to $z_2=k_2$ in $G$, which is a contradiction. Therefore, $H'$ cannot be a subdivision of $K_{2, 3}$.

Since, $G'$ does not contain a subgraph $H'$ that is a subdivision of $K_4$ or $K_{2,3}$, $G'$ is outerplanar.
\qed
\end{proof}
The following lemma explains the effect of the addition of an edge $(u, v)$ as mentioned in Lemma \ref{extensionlemma}, to the block structure and the Hamiltonian cycle of each block. Assume that for $0\le i \le k$, the edge $(x_i, x_{i+1})$ belongs to the block $B_i$. 
\begin{lemma}\label{Hamiltonianunchanged}
\begin{enumerate}
 \item Other than the blocks $B_0$ to $B_k$ of $G$ merging together to form a new block $B'$ of $G'$, blocks in $G$ and $G'$ are the same.
 \item Vertices in blocks $B_0$ to $B_k$, except $x_i$, $0 \le i \le k+1$, retains their successor and predecessor in the Hamiltonian cycle of $B'$ same as it was in its respective block's Hamiltonian cycle in $G$.
 \item Each $x_i$, $0 \le i \le k$, retains its Hamiltonian cycle predecessor in $B'$ same as it was in the block $B_i$ of $G$ and each $x_i$, $1 \le i \le k+1$, retains its Hamiltonian cycle successor in $B'$ same as in the block $B_{i-1}$ of $G$. 
\end{enumerate}
\end{lemma}
\begin{proof}
When the edge $(u, v)$ is added, it creates a cycle containing the vertices $u=x_0, x_1,\ldots, x_{k+1}=v$. Hence, the blocks $B_0$ to $B_k$ of $G$ merge together to form a single block $B'$ in $G'$. It is obvious that other blocks are unaffected by this edge addition. 

For simplicity, if $B_i$ is a trivial block containing the edge $(x_i, x_{i+1})$, we say that $x_i$ and $x_{i+1}$ are neighbors of each other in the Hamiltonian cycle of $B_i$. For each $B_i$, $0 \le i \le k$, let $x_i, x_{i+1}, z_{i1}, z_{i2}, \ldots, z_{it_i}, x_i$ be the Hamiltonian cycle of $B_i$ in $G$. For $0 \le i \le k$, let us denote the path $x_{i+1}, z_{i1}, z_{i2}, \ldots, z_{it_i}$ by $P_i$. Then, the Hamiltonian cycle of $B'$ is $u=x_0 \circ P_k \circ P_{k-1}\circ \ldots P_0 \circ u$, where $\circ$ denotes the concatenation of the paths. (For example, in Fig. \ref{fig1}, $u, v, a, b, c, d, e, f, g, h, i, j, k, l, u$ is the Hamiltonian cycle of the new block formed, when the edge $(u, v)$ is added.) From this, we can conclude that the second and third parts of the lemma holds.
\qed
\end{proof}
\section{Stage 2: Construction of $G'$ and its path decomposition}
The organization of this sections is as follows: For each cut vertex $x$ of $G$, we define an ordering among the child blocks attached through $x$ to their parent block, based on the nice path decomposition $(\mathcal{P}, \mathscr{X})$ of $G$ obtained using Lemma \ref{nicepath}. This ordering is then used in defining a supergraph $G'(V, E')$ of $G$ such that for every cut vertex $x$ in $G'$, $G' \setminus x$ has exactly two components. Using repeated applications of Lemma \ref{extensionlemma}, we then show that $G'$ is outerplanar. We extend the path decomposition $(\mathcal{P}, \mathscr{X})$ of $G$ to a path decomposition $(\mathcal{P}', \mathscr{X}')$ of $G'$, as described in Section \ref{background}. By a counting argument using the properties of the nice path decomposition $(\mathcal{P}, \mathscr{X})$, we show that the width of the path decomposition $(\mathcal{P}', \mathscr{X}')$ is at most $2p'+1$, where $p'$ is the width of $(\mathcal{P}, \mathscr{X})$.
\subsection{Defining an ordering of child blocks}\label{ordering}
 If $(\mathcal{P}, \mathscr{X})$ is a nice path decomposition of $G$, then, for each non-root block $B$ of $G$, at least one bag in $\mathscr{X}$ contains both the first and last vertices of $B$ together. 
\begin{definition}[Sequence number of a non-root block]
Let $(\mathcal{P}, \mathscr{X})$ be the nice path decomposition of $G$ obtained using Lemma \ref{nicepath}. For each non-root block $B$ of $G$, we define the sequence number of $B$ as $\min\{i \mid X_i \in \mathscr{X}$ simultaneously contains both the first and last vertices of $B\}$.
\end{definition}
For each cut vertex $x$, there is a unique block $B^x$ such that $B^x$ and its child blocks are intersecting at $x$. For each cut vertex $x$, we define an ordering among the child blocks attached at $x$, as follows. If $B_1,\ldots,B_k$ are the child blocks attached at $x$, we order them in the increasing order of their sequence numbers in $(\mathcal{P}, \mathscr{X})$. If $B_i$ and $B_j$ are two child blocks with the same sequence number, their relative ordering is arbitrary. 

From the ordering defined, we can make some observations about the appearance of the first and last vertices of a block $B_i$ in the path decomposition. These observations are crucially used for bounding the width of the path decomposition $(\mathcal{P}', \mathscr{X}')$ of $G'$. Let $B_1,\ldots, B_k$ be the child blocks attached at a cut vertex $x$, occurring in that order according to the ordering we defined above. For $1 \le i \le k$, let $y_i$ and $y'_i$ respectively be the last and first vertices of $B_i$.
\begin{property}\label{prop1}
 For any $1 \le i \le k-1$, if $Gap_\mathscr{X}(y'_i, y_{i+1})\ne \emptyset$, then\\$Gap_\mathscr{X}(y'_i, y_{i+1})=[LastIndex_\mathscr{X}(y'_i)+1, FirstIndex_\mathscr{X}(y_{i+1})]$ and $x \in X_t$ for all $t \in Gap_\mathscr{X}(y'_i, y_{i+1})$.
\end{property}
\begin{proof}
 If  $Gap_\mathscr{X}(y'_i, y_{i+1})\ne \emptyset$, either $LastIndex_\mathscr{X}(y'_i)<FirstIndex_\mathscr{X}(y_{i+1})$ or $LastIndex_\mathscr{X}(y_{i+1})<FirstIndex_\mathscr{X}(y'_{i})$. The latter case will imply that, sequence number of $B_{i+1}<$ sequence number of $B_i$, which is a contradiction. Therefore, $LastIndex_\mathscr{X}(y'_i)<FirstIndex_\mathscr{X}(y_{i+1})$ and hence $Gap_\mathscr{X}(y'_i, y_{i+1})=[LastIndex_\mathscr{X}(y'_i)+1, FirstIndex_\mathscr{X}(y_{i+1})]$.
 
 Since $x$ is adjacent to $y'_i$ and $y_{i+1}$, we get $FirstIndex_\mathscr{X}(x)\le LastIndex_\mathscr{X}(y'_i)$ and $LastIndex_\mathscr{X}(x) \ge$\\$FirstIndex_\mathscr{X}(y_{i+1})$. We can conclude that $Range_\mathscr{X}(x) \supseteq [LastIndex_\mathscr{X}(y'_i), FirstIndex_\mathscr{X}(y_{i+1})]$ and the property follows. 
\qed
\end{proof}
\begin{property}\label{prop2}
 For any $1 \le i <j \le k-1$, $Gap_\mathscr{X}(y'_i, y_{i+1}) \cap Gap_\mathscr{X}(y'_j, y_{j+1})=\emptyset$. 
\end{property}
\begin{proof}
We can assume that $Gap_\mathscr{X}(y'_i, y_{i+1})\ne \emptyset$ and $Gap_\mathscr{X}(y'_j, y_{j+1}) \ne \emptyset$, since the property holds trivially otherwise. By Property \ref{prop1}, we get, $Gap_\mathscr{X}(y'_i, y_{i+1})=[LastIndex_\mathscr{X}(y'_i)+1, FirstIndex_\mathscr{X}(y_{i+1})]$ and $Gap_\mathscr{X}(y'_j, y_{j+1})=$\\$[LastIndex_\mathscr{X}(y'_j)+1, FirstIndex_\mathscr{X}(y_{j+1})]$. Since $i+1 \le j$, by the property of the ordering of blocks, we know that sequence number of $B_{i+1} \le$ sequence number of $B_j$. From the definitions, we have,  $FirstIndex_\mathscr{X}(y_{i+1})\le$ sequence number of $B_{i+1} \le $sequence number of $B_j \le LastIndex_\mathscr{X}(y'_j)$ and the property follows.
\qed
\end{proof}
\subsection{Algorithm for constructing $G'$ and its path decomposition}
We use Algorithm \ref{algMulticut} to construct $G'(V, E')$ and a path decomposition $(\mathcal{P}', \mathscr{X}')$ of $G'$. The processing of each cut vertex is done in lines \ref{loop1} to \ref{modify} of Algorithm \ref{algMulticut}. While processing a cut vertex $x$, the algorithm adds the edges $(y_1', y_2), (y_2', y_3), \ldots,$  $(y'_{k_x-1}, y_{k_x})$ (as defined in the algorithm) and modifies the path decomposition, to reflect each edge addition. 
\begin{algorithm}
     \LinesNumbered 
     \DontPrintSemicolon
    \caption{Computing the intermediate supergraph $G'$ and its path decomposition}
     \label{algMulticut}    
     \KwIn{A connected outerplanar graph $G(V$, $E)$ and a nice path decomposition $(\mathcal{P}, \mathscr{X})$ of $G$, 
           the rooted block tree of $G$, the Hamiltonian cycle of each non-trivial block of $G$ and the first and last vertices of each non-root block of $G$}
     \KwOut{An outerplanar supergraph $G'(V$, $E')$ of $G$ such that, for every cut vertex $x$ of $G'$, $G'\setminus x$ has exactly two connected components, a path decomposition $(\mathcal{P}', \mathscr{X}')$ of $G'$}    
    $E'=E$, $(\mathcal{P}', \mathscr{X}')= (\mathcal{P}, \mathscr{X})$\;
     \For {each cut vertex $x \in V(G)$\label{loop1}}{
          Let $B_1,\ldots,B_{k_x}$, in that order, be the child blocks attached at $x$, according to the ordering defined in Section \ref{ordering}\;
          \For {$i=1$ to $k_x-1$}{
               Let $y'_i$ be the first vertex of $B_i$ and $y_{i+1}$ be the last vertex of $B_{i+1}$\;
               $E' = E' \cup \{(y'_i, y_{i+1})\}$\;
               \lIf{$Gap_\mathscr{X}(y'_i, y_{i+1})\ne \emptyset$}{\label{loop2}    
               \lFor{$t\in Gap_\mathscr{X}(y'_i, y_{i+1})$}{
                    $X'_t = X'_t \cup \{y'_i\}$\label{modify}
               }
               }
          }
}
\end{algorithm}
\begin{lemma}\label{outerplanarfirst}
   $G'$ is outerplanar and for each cut vertex $x$ of $G'$, $G'\setminus x$ has exactly two components.
\end{lemma}
\begin{proof}
 We know that $G$ is outerplanar to begin with. At a certain stage, let $x$ be the cut vertex taken up by the algorithm for processing (in line \ref{loop1}). Assume that the graph at this stage, denoted by $G_0$, is outerplanar and each cut vertex $x'$ whose processing is completed, satisfies the condition that all the child blocks attached at $x'$ have merged together to form a single child block attached at $x'$. 

It is clear that the child blocks attached at a vertex $x$ remain unchanged until $x$ is picked up by the algorithm for processing. Let $B_1,\ldots,B_{k_x}$, in that order, be the child blocks attached at $x$, according to the ordering defined in Section \ref{ordering}. Let $B^x$ be the parent block of $B_1,\ldots,B_{k_x}$, in the current graph $G_0$. For each $1 \le i \le k_x$, let $y'_i$ and $y_i$ respectively be the first and last vertices of $B_i$. For $1 \le i \le k_x-1$, let $G_{i}$ be the graph obtained, when the algorithm has added the edges up to $(y'_i, y_{i+1})$.

We will prove that the algorithm maintains the following invariants, while processing the cut vertex $x$, for each $0 \le i \le k_x-1$:
\begin{quote}
\textit{The graph $G_{i}$ is outerplanar. In $G_{i}$, the blocks $B_1,\ldots,B_{i+1}$ of $G_{i-1}$ have merged together and formed a child block $B'$ of $B^x$. The vertex $y'_{i+1}$ is the first vertex of $B'$. If $i\le k_x-2$, blocks $B_{i+2},\ldots,B_{k_x}$ remain the same in $G_i$, as in $G$.}
\end{quote}
By our assumption, the invariants hold for $G_0$. We need to show that if the invariants hold for $G_{i-1}$, they hold for $G_{i}$ as well. Assume that the invariants hold for $G_{i-1}$ and let $B'$ be the child block of $B^x$ in $G_{i-1}$ that is formed by merging together the blocks $B_1,\ldots,B_{i}$ of $G_{i-2}$, as stated in the invariant. Since the invariants hold for $G_{i-1}$ by our assumption, $y'_{i}$ is the first vertex of $B'$ and $y_{i+1}$ is the last vertex in $B_{i+1}$. In other words, $y'_{i}$ is the successor of $x$ in $B'$ and $y_{i+1}$ is the predecessor of $x$ in $B_{i+1}$ and the edges $(y_{i+1}, x)$ and $(x, y'_i)$ of the path $P_i =(y_{i+1}, x, y_i')$ belong to two different blocks of $G_{i-1}$. Hence, by Lemma \ref{extensionlemma}, after adding the edge $(y_i', y_{i+1})$, the resultant intermediate graph $G_i$ is outerplanar. By Lemma \ref{Hamiltonianunchanged}, the blocks $B'$ and $B_{i+1}$ merges together to form a child block $B''$ of $B^x$ in $G_i$. Further, the vertex $y'_{i+1}$ 
will be the successor of $x$ in the Hamiltonian cycle of $B''$ i.e, the first of the block $B''$. Remaining blocks of $G_i$ are the same as in $G_{i-1}$. Thus, all the invariants hold for $G_i$. It follows that the graph $G_{k_x-1}$ is outerplanar and the blocks $B_1,\ldots,B_{k_x}$ have merged 
together in $G_{k_x-1}$ to form a single child block of $B^x$ at $x$. 

When this processing is repeated at all cut vertices, it is clear that $G'$ is outerplanar and for each cut vertex $x$ of $G'$, $G'\setminus x$ has exactly two components.
\qed
\end{proof}
\begin{lemma}
 $(\mathcal{P}', \mathscr{X}')$ is a path decomposition of $G'$ of width at most $8p+7$. 
\end{lemma}
\begin{proof}
 Algorithm \ref{algMulticut} initialized $(\mathcal{P}', \mathscr{X}')$ to $(\mathcal{P}, \mathscr{X})$ and modified it in line \ref{modify}, following each edge addition. By Property \ref{prop1}, we have $Gap_\mathscr{X}(y'_i, y_{i+1})=[LastIndex_\mathscr{X}(y'_i)+1, FirstIndex_\mathscr{X}(y_{i+1})]$. Hence, by the modification done in line \ref{modify} while adding a new edge $(y'_i, y_{i+1})$, $(\mathcal{P}', \mathscr{X}')$ becomes a path decomposition of the graph containing the edge $(y'_i, y_{i+1})$, as explained in Section \ref{background}. It follows that, when the algorithm terminates $(\mathcal{P}', \mathscr{X}')$ is a path decomposition of $G'$. 

Consider any $X'_t \in \mathscr{X}'$. While processing the cut vertex $x$, if Algorithm \ref{algMulticut} inserts a new vertex $y'_i$ to $X'_t$, to reflect the addition of a new edge $(y'_i, y_{i+1})$ then, $t \in Gap_\mathscr{X}(y'_i, y_{i+1})$. Suppose $(y'_i, y_{i+1})$ and $(y'_j, y_{j+1})$ are two new edges added while processing the cut vertex $x$, where, $1 \le i <j \le k_x-1$. By property \ref{prop2}, we know that if $t \in Gap_\mathscr{X}(y'_i, y_{i+1})$, then, $t \notin Gap_\mathscr{X}(y'_j, y_{j+1})$. Therefore, when the algorithm processes a cut vertex $x$ in lines \ref{loop1} to \ref{modify}, at most one vertex is newly inserted to the bag $X'_t$. Moreover, if $t \in Gap_\mathscr{X}(y'_i, y_{i+1})$ then, the cut vertex $x \in X_t$, by Property \ref{prop1}. 

That means, a vertex not present in the bag $X_t$ can be added to $X'_t$ only when a cut vertex $x$ that is already present in the bag $X_t$ is being processed. Moreover, when a cut vertex $x$ that is present in $X_t$ is processed, at most one new vertex can be added to $X'_t$. It follows that $|X'_t| \le |X_t|+$number of cut vertices present in $X_t \le 2 |X_t| \le 2(4p+4)$. Therefore, the width of the path decomposition $(\mathcal{P}', \mathscr{X}')$ is at most $8p+7$.
\qed
\end{proof}
\section{Construction of $G''$ and its path decomposition}
In this section, we give an algorithm to add some more edges to $G'(V, E')$ so that the resultant graph $G''(V, E'')$ is $2$-vertex-connected. The algorithm also extend the path decomposition $(\mathcal{P}', \mathscr{X}')$ of $G'$ to a path decomposition $(\mathcal{P}'', \mathscr{X}'')$ of $G''$. The analysis of the algorithm shows the existence of a surjective mapping from the cut vertices of $G'$ to the edges in $E'' \setminus E'$. A counting argument based on this surjective mapping shows that the width of the path decomposition $(\mathcal{P}'', \mathscr{X}'')$ is at most $16p+15$. For making our presentation simpler, if a block $B_i$ is just an edge $(u,v)$, we abuse the definition of a Hamiltonian cycle and say that $u$ and $v$ are clockwise neighbors of each other in the Hamiltonian cycle of $B_i$.

Recall that for every cut vertex $x$ of $G'$, $G'\setminus x$ has exactly two components. Since any cut vertex belongs to exactly two blocks of $G$, based on the rooted block tree structure of $G$, we call them the parent block containing $x$ and the child block containing $x$. We use $child_{x}(B)$ to denote the unique child block of the block $B$ at the cut vertex $x$ and $parent(B)$ to denote the parent block of the block $B$. For a block $B$, $next_B(v)$ denotes the successor of the vertex $v$ in the Hamiltonian cycle of $B$. We say that a vertex $u$ is \textit{encountered by the algorithm}, when $u$ gets assigned to the variable $v'$, during the execution of the algorithm. The block referred to by the variable $B$ represents the \textit{current block} being traversed.

To get a high level picture of our algorithm, the reader may consider it as a traversal of vertices of $G'$, starting from a non-cut vertex in the root block of $G'$ and proceeding to the successor of $v$ on reaching a non-cut vertex $v$. On reaching a cut vertex $x$, the algorithm bypasses $x$ and recursively traverses the child block containing $x$ and its descendant blocks, starting from the successor of $x$ in child block containing $x$. After this, the algorithm comes back to $x$ to visit it, and continues the traversal of the remaining graph, by moving to the successor of $x$ in the parent block containing $x$. Before starting the recursive traversal of the child block containing $x$ and its descendant blocks, the algorithm sets $bypass(x)=TRUE$. (Note that, since there is only one child block attached to any cut vertex, each cut vertex is bypassed only once.) In this way, when a sequence of one or more cut vertices is bypassed, an edge is added from the vertex visited just before bypassing the first 
cut vertex in the sequence to the vertex visited just after bypassing the last cut vertex in the sequence. The path decomposition is also modified, to reflect this edge addition. The detailed algorithm to $2$-vertex-connect $G'$ is given in Algorithm \ref{algBiConnect}.

\begin{algorithm}[h]
     \LinesNumbered 
     \DontPrintSemicolon
    \caption{Computing a $2$-vertex-connected outerplanar supergraph}
     \label{algBiConnect}    
     \KwIn{A connected outerplanar graph $G'(V$, $E')$ such that $G'\setminus x$ has exactly two connected components for every cut vertex $x$ of $G'$. A path decomposition $(\mathcal{P}', \mathscr{X}')$ of $G'$. The rooted block tree of $G'$, the Hamiltonian cycle of each non-trivial block of $G'$ and the first and last vertices of each non-root block of $G'$}
     \KwOut{A $2$-vertex-connected outerplanar supergraph $G''(V$, $E'')$ of $G'$, a path decomposition $(\mathcal{P}'', \mathscr{X}'')$ of $G''$}    
     $E'' =E'$, $(\mathcal{P}'', \mathscr{X}'')$ = $(\mathcal{P}', \mathscr{X}')$\;
     \For{each vertex $v \in V(G')$ \label{forLoop}}{
         completed($v$) = FALSE, 
         \lIf {$v$ is a cut vertex}{
         bypass($v$) = FALSE 
         }
     }
     $B$ = rootBlock \label{rootBLine} \;
     Choose $v'$ to be some non-cut vertex of the rootBlock \label{initialVertexLine} \;
     $completed(v')=$TRUE, $completedCount = 1$ \label{completeLine1} \;
     $v=v'$ \label{reassignline1} \;
     \While {$completedCount < |V(G')|$ \label{outerWhile} }{
         $v' = next_B(v)$ \label{reassignv'} \;
         $bypassLoopTaken=$FALSE,  $sequence=$ EmptyString \label{initialSeq} \;
	\While{ $v'$ is a cut vertex and $bypass(v')$ is FALSE \label{innerWhile} }{
            $bypassLoopTaken=$TRUE\;  
            $bypass(v')=$TRUE,  $sequence=$Concatenate$(sequence, v')$ \label{bypassline} \;
            $B = child_{v'}(B)$, $v'=next_B(v')$ \label{childBlockLine} \; 
         }
         \If {$bypassLoopTaken$ is TRUE \label{checkLoopTaken}}{
             $e= (v, v')$,  $bypassSeq(e)=sequence$ \label{seqLine} \;
             $E'' = E'' \cup \{e\}$\label{edgeline}\;
             \If{$Gap_\mathscr{X'}(v, v')\ne \emptyset$}{
                       \lIf{$LastIndex_\mathscr{X'}(v)<FirstIndex_\mathscr{X'}(v')$}{
                           \lFor{$t\in Gap_\mathscr{X'}(v, v')$}{  
                            $X''_t = X''_t \cup \{v\}$ 
                         } 
                        }\;   
                        \lElseIf{$LastIndex_\mathscr{X'}(v')<FirstIndex_\mathscr{X'}(v)$}{
                            \lFor{$t\in Gap_\mathscr{X'}(v, v')$}{
                               $X''_t = X''_t \cup \{v'\}$ 
                             }
                        }
                     }

        }
        \If{$v'$ is a cut vertex and $bypass(v')$ is TRUE \label{checkLine}}{
              $B = parent(B)$ \label{parentBlockLine}
         }
        $completed(v')$= TRUE, $completedCount = completedCount +1$ \label{completeLine2}\; 
        $v=v'$\label{reassignline}\;   
     }        
\end{algorithm}
If $G'$ has only a single vertex, then it is easy to see that the algorithm does not modify the graph. For the rest of this section, we assume that this is not the case. The following recursive definition is made in order to make the description of the algorithm easier. 
\begin{definition}
Let $G$ be a connected outerplanar graph with at least two vertices such that $G\setminus x$ has exactly two connected components for every cut vertex $x$ and $v$ be a non-cut vertex in the root block of $G$. For any cut vertex $x$ of $G$, let $G_{x}$ denote the subgraph of $G$ induced on the vertices belonging to the unique child block attached at $x$ and all its descendant blocks. If the root-bock of $G$ is a non-trivial block, let $v, v_1, \ldots, v_t, v$ be the Hamiltonian cycle of the root-block of $G$, starting at $v$. Then,
\begin{equation*}
Order(G, v) =
\begin{cases}  
v, v_1 &\text{if $G$ is a single edge $(v, v_1)$}\\
v, v_1, \ldots, v_t &\text{if $G$ is $2$-vertex-connected}\\
v, S_1, \ldots, S_t &\text{otherwise, where for $1 \le i \le t$, $S_i =$} \begin{cases}v_i &{\text{ if $v_i$ is not a cut vertex in $G$}}\\Order(G_{v_i}, v_i),v_i & \text{ otherwise.} \end{cases} 
\end{cases}
\end{equation*}
\end{definition}
The following lemma gives a precise description of the order in which the algorithm encounters vertices of $G'$.
\begin{lemma}\label{orderEncounter}
Let $G'$ be a connected outerplanar graph with at least two vertices, such that for every cut vertex $x$ of $G'$, $G'\setminus x$ has exactly two connected components. If $G'$ is given as the input graph to Algorithm \ref{algBiConnect} and $v_0$ is the non-cut vertex in the root block of $G'$ from which the algorithm starts the traversal, then $Order(G', v_0)$ is the order in which Algorithm \ref{algBiConnect} encounters the vertices of $G'$. 
\end{lemma}
Since the proof of this lemma is easy but is lengthy and detailed, in order to make the reading easier we have moved the proof to the Appendix. Now we will use Lemma \ref{orderEncounter} and derive some properties maintained by Algorithm \ref{algBiConnect}. 
\begin{property}\label{completionProperty}
\begin{enumerate}
\item Every non-cut vertex of $G'$ is encountered exactly once. Every cut vertex of $G'$ is encountered exactly twice. 
\item A non-cut vertex is completed when it is encountered for the first time. A cut vertex is bypassed when it is encountered for the first time and is completed when it is encountered for the second time. Each cut vertex is bypassed exactly once. 
\item Every vertex is completed exactly once and a vertex that is declared completed is never encountered again. 
\end{enumerate}
\end{property}
\begin{proof}
The first part of the property follows directly from Lemma \ref{orderEncounter}. 

To prove the second part, observe that a vertex $u$ is encountered only in Lines \ref{initialVertexLine}, \ref{reassignv'} or \ref{childBlockLine}. If $v'=u$ is a non-cut vertex, the inner while-loop will not be entered. The vertex $u$ is completed in Line \ref{completeLine1} or Line \ref{completeLine2} before another vertex is encountered by executing Line \ref{reassignv'} again.

For any cut vertex $x$, $bypass(x)=$ FALSE initially and it is changed only after $x$ is encountered and the algorithm enters the inner while-loop with $v'=x$. When $x$ is first encountered, $bypass(x)$ is FALSE and the algorithm gets into the inner while-loop and inside this loop $bypass(x)$ is set to TRUE. After this, $bypass(x)$ is never set to FALSE. Therefore, when $x$ is encountered for the second time, the inner while-loop is not entered and before $v'$ gets reassigned, $x$ is completed in Line \ref{completeLine2}. 

The third part of the property follows from the first two parts and Lemma \ref{orderEncounter}.
\qed
\end{proof}
\begin{lemma}\label{surjectionLemma} 
Each cut vertex of $G'$ is bypassed exactly once by the algorithm and is associated with a unique edge in $E''\setminus E'$. Every edge $e \in E'' \setminus E'$ has a non-empty sequence of bypassed cut vertices associated with it, given by $bypassSeq(e)$. Hence, the function $f :$ cut vertices of $G' \mapsto E'' \setminus E'$, defined as
\begin{equation*}
 f(x)=e\text{ such that }x\text{ is present in }bypassSeq(e)
\end{equation*} 
is a surjective map.
\end{lemma}
\begin{proof}
From Property \ref{completionProperty}, each cut vertex $x$ of $G'$ is bypassed exactly once by the algorithm. Note that when $x$ is bypassed in Line \ref{bypassline}, $x$ is appended to the string $sequence$ and the variable $bypassLoopTaken$ was set to TRUE in the previous line. On exiting the inner while-loop, since $bypassLoopTaken$=TRUE, an edge is added in Line \ref{edgeline}. Before adding the new edge $e$, in the previous line the algorithm set $bypassSeq(e)$ to be the sequence of cut vertices accumulated in the variable $sequence$. As we have seen, $x$ is present in the string $sequence$ and it is clear from the algorithm that $sequence$ was not reset to emptystring before assigning it to $bypassSeq(e)$. Therefore, $x$ is present in $bypassSeq(e)$. In the next iteration of the outer while-loop, $sequence$ is reset to emptystring. Since $x$ is bypassed only once, it will not be added to the string $sequence$ again nor it will be part of $bypassSeq(e')$ for any other edge $e'$. 

An edge $e$ gets added in Line \ref{edgeline} only if $bypassLoopTaken$ is TRUE, while executing Line \ref{checkLoopTaken}. However, the variable $bypassLoopTaken$ is set to FALSE in the outer while-loop every time just before entering the inner while-loop and is set to TRUE only inside the inner while-loop, where at least one cut vertex is bypassed in Line \ref{bypassline} and is added to $sequence$. Until the algorithm exits from the inner while-loop, and the next edge $e$ is added, $bypassLoopTaken$ is maintained to be TRUE. This ensures that $bypassSeq(e)$ is a non-empty string for each edge $e \in E'' \setminus E'$. 
\qed
\end{proof}
\begin{property}\label{pathBetweenPairs}
Let $e=(u_i, v_i)$ be an edge such that, at the time of adding the edge $e$ in Line \ref{edgeline} of Algorithm \ref{algBiConnect}, the variable $v$ contained the value $u_i$ and the variable $v'$ contained the value $v_i$. 
\begin{itemize}
\item The vertex $v=u_i$ is already completed at this time and the vertex $v_i$ is completed subsequently. 
\item Cut vertices that belong to $bypassSeq(e)$ are precisely the vertices bypassed during the period from the execution of Line \ref{reassignv'} just before adding the edge $e$ in Line \ref{edgeline} to the time when $e$ is added in Line \ref{edgeline}. These vertices were bypassed in the order in which they appear in $bypassSeq(e)$.
\item Each cut vertex bypassed before bypassing the first cut vertex that belong to $bypassSeq(e)$ belongs to the bypass sequence of one of the edges in $E'' \setminus E'$ which was added before $e$.
\item If $bypassSeq(e)=x_1, x_2,\ldots, x_k$, then $u_i=x_0, x_1, x_2, \ldots, x_k, x_{k+1}=v_i$ is a path in $G'$ such that 
\begin{itemize}
\item for each $1 \le i \le k$, $x_{i}$ is the successor of $x_{i-1}$ in the Hamiltonian cycle of the parent block in $G'$, at the cut vertex $x_i$.
\item $v_i=x_{k+1}$ is the successor of $x_{k}$ in the Hamiltonian cycle of the child block in $G'$, at the cut vertex $x_k$.
\end{itemize}
Since each bypassed vertex is a cut vertex in $G'$, it is easy to see that $e=(u_i, v_i)$ was not already an edge in $G'$.
\end{itemize}
\end{property}
\begin{proof}
Suppose that at the time of adding the edge $e$ in Line \ref{edgeline} of Algorithm \ref{algBiConnect}, the variable $v$ contained the value $u_i$ and the variable $v'$ contained the value $v_i$.
Observe that the variable $v$ always gets its value from the variable $v'$ in lines \ref{reassignline1} and \ref{reassignline} and just before this, in Lines \ref{completeLine1} and \ref{completeLine2} $v'$ was declared completed. Therefore the vertex assigned to $v$ is always a completed vertex. Therefore, at the time of adding the edge $e$ in Line \ref{edgeline}, the vertex $v=u_i$ is already completed. After the edge is added in Line \ref{edgeline}, in Line \ref{completeLine2} the vertex assigned to $v'=v_i$ is completed. Since by Property \ref{completionProperty} a vertex is completed only once, this is the only time at which $v_i$ is completed. 

Since each cut vertex is bypassed only once, and is added to the bypass sequence of the next edge added (see the proof of Lemma \ref{surjectionLemma}), it is clear that if a cut vertex is bypassed before $x_1$, it should belong to the bypass sequence of an edge in $E''\setminus E'$ added before $e$. The remaining parts of the property are easy to deduce from lines \ref{innerWhile} - \ref{childBlockLine} and Line \ref{seqLine} of the algorithm.
\qed
\end{proof}
\begin{lemma}\label{biconnect}
  $G''$ is $2$-vertex-connected.
\end{lemma}
\begin{proof}
  We show that $G''$ does not have any cut vertices. Since $G''$ is a supergraph of $G'$, if a vertex $x$ is not a cut vertex in $G'$, it will not be a cut vertex in $G''$. We need to show that the cut vertices in $G'$ become non-cut vertices in $G''$. Consider a newly added edge $(u, v)$ of $G''$. Without loss of generality, assume that $u$ was completed before $v$ in the traversal and for $e=(u, v)$, $bypassSeq(e)=(x_1,x_2, \ldots, x_k)$. By Property \ref{pathBetweenPairs}, $u, x_1, x_2, \ldots, x_k, v$ is a path in $G'$. When our algorithm adds the edge $(u, v)$, it creates the cycle $u, x_1, x_2, \ldots, x_k, v, u$ in the resultant graph. Recall that, for each $1 \le i \le k$, $G'\setminus x_i$ had exactly two components; one containing $x_{i-1}$ and the other containing $x_{i+1}$. After the addition of the edge $(u, v)$, vertices $x_{i-1}$, $x_i$ and $x_{i+1}$ lie on a common cycle. Hence, after the edge $(u, v)$ is added, for $1 \le i \le k$, $x_i$ is no longer a cut vertex. Since by Lemma \ref{surjectionLemma} every cut vertex in $G'$ was part of the bypass sequence associated with some edge in $E'' \setminus E'$, all of them become non-cut vertices in $G''$. 
\qed
\end{proof}
To prove that $G''$ is outerplanar, we can imagine the edges in $E'' \setminus E'$ being added to $G'$ one at a time. Our method is to repeatedly use Lemma \ref{extensionlemma} and show that after each edge addition, the resultant graph remains outerplanar. We will first note down some properties maintained by Algorithm \ref{algBiConnect}.

Let $\{e_i= (u_i, v_i)\mid 1 \le i \le m =|E''\setminus E'|\}$ be the set of edges added by Algorithm \ref{algBiConnect}. Assume that, for each $1 \le i< m$, $(u_i, v_i)$ was added before $(u_{i+1}, v_{i+1})$ and at the time of adding the edge $e_i$ in Line \ref{edgeline} of Algorithm \ref{algBiConnect}, the variable $v$ contained the value $u_i$ and the variable $v'$ contained the value $v_i$. By Property \ref{pathBetweenPairs}, $u_i$ is completed before $v_i$. Let $bypassSeq((u_i, v_i))=x^i_1, x^i_2, \ldots, x^i_{k_i}$, where $k_i\ge 1$, and $P^i=(u_i = x^i_0 , x^i_1 , x^i_2 , . . . , x^i_{k_i} , x^i_{k_i+1} = v_i)$ be the associated path in $G'$ (Property \ref{pathBetweenPairs}). Let $B^i_j$ denote the block containing the edge $(x^i_j,x^i_{j+1})$ in $G'$. Clearly, $B^1_0$ is the root block of $G'$. The following statement is an immediate corollary of Property \ref{pathBetweenPairs}, with the definitions above.
\begin{property}\label{prop3}
For each $0 \le j \le k_i$, the vertex $x^i_{j+1}$ is the successor of the vertex $x^i_j$ in the Hamiltonian cycle of the block $B^i_{j}$. The path $P^i$ shares only one edge with any block of $G'$.   
\end{property}
\begin{property}\label{initialvertex}
 If $1 \le i < j \le m = |E''\setminus E'|$, then $u_i \ne u_j$. 
\end{property}
\begin{proof}
By our assumption, at the time of adding the edge $(u_i, v_i)$, we had $v=u_i$ and $v'=v_i$. By Property \ref{pathBetweenPairs}, the vertex $v=u_i$ was already completed at this time. After adding the edge $(v, v')=(u_i, v_i)$, the algorithm reassigns $v=v'=v_i$ in Line \ref{reassignline}. By Property \ref{completionProperty}, the algorithm will never encounter the completed vertex $u_i$ again, and this means that $v'$ is never set to $u_i$ in future. This also implies that $v$ is never set to $u_i$ in future, since $v$ gets reassigned later only in Line \ref{reassignline}, where it gets its value from the variable $v'$. Since $v=u_j$ when the edge $(u_j, v_j)$ is added, we have $u_i \ne u_j$. 
\qed
\end{proof}
At a stage of Algorithm \ref{algBiConnect}, we say that a non-root block $B$ is \textbf{touched}, if at that stage Algorithm \ref{algBiConnect} has already bypassed the cut vertex $y$ such that $B$ is the child block containing $y$. At any stage of Algorithm \ref{algBiConnect}, we consider the root block of $G'$ to be touched. 
\begin{property}\label{touchedblocks}
 When Algorithm \ref{algBiConnect} has just finished adding the edge $(u_i, v_i)$, the touched blocks are precisely $B^1_0 \cup \bigcup_{1 \le j \le i}\{B^j_1, \ldots, B^j_{k_j}\}$. This implies that if $i<m$, the blocks $\{B^{i+1}_1, B^{i+1}_2, \ldots, B^{i+1}_{k_i}\}$ remain untouched when the algorithm has just finished adding the edge $(u_{i}, v_{i})$. 
\end{property}
\begin{proof}
The root block $B^1_0$ is always a touched block by definition and the other touched blocks when Algorithm \ref{algBiConnect} has just finished adding the edge $(u_i, v_i)$ are the child blocks attached to cut vertices bypassed so far. However, by Property \ref{pathBetweenPairs}, when Algorithm \ref{algBiConnect} has just finished adding the edge $(u_i, v_i)$, a cut vertex $x$ is already bypassed if and only if $x$ belongs to $bypassSeq(e_j)$ for some $j \le i$. For $j \le i$, we had $bypassSeq((u_j, v_j))=x^j_1, x^j_2, \ldots, x^j_{k_j}$ and for $1\le l\le k_j$, $B^j_l$ is the child block attached at $x^j_l$ by the last part of Property \ref{pathBetweenPairs}. Hence, the initial part of the property holds. From this, the latter part of the property follows, by Lemma \ref{surjectionLemma}.
\qed
\end{proof}
\begin{property}\label{blocksconnected}
 For each $2 \le i \le m$, when the algorithm has just finished adding the edge $(u_{i-1}, v_{i-1})$, the block $B^i_0$ is a touched block. 
\end{property}
\begin{proof}
If $B^i_0$ is the root block, the property is trivially true. Assume that this is not the case. 

Consider the situation when the algorithm has just finished adding the edge $(u_{i-1}, v_{i-1})$ in Line \ref{edgeline}. By Property \ref{pathBetweenPairs}, the next cut vertex to be bypassed is $x^i_1$, which is the first vertex appearing in $bypassSeq(e_i)$, and $B^i_0$ is the parent block attached at $x^i_1$. 
Let $y$ be the cut vertex such that $B^i_0$ is the child block at $y$. If $y$ has been encountered by now, $y$ would have been bypassed (Lemma \ref{completionProperty}), making $B^i_0$ a touched block. Since a cut vertex is bypassed when it is encountered for the first time (Lemma \ref{completionProperty}) and $y$ is the first vertex the algorithm encounters among the vertices in the block $B^i_0$ (Lemma \ref{orderEncounter}), if $y$ is not yet encountered, it will contradict the fact that $x^i_1$ is the next cut vertex to be bypassed, because $x^i_1$ is a vertex in $B^i_0$. Therefore $y$ should have been encountered earlier and therefore, $B^i_0$ is a touched block.
\qed
\end{proof}
\begin{lemma}
 $G''$ is outerplanar. 
\end{lemma}
\begin{proof}
Let $G'_0=G'$ and for each $1 \le i \le m$, let $G'_i(V, E'_i)$ be the graph obtained by assigning $E'_i=E' \cup \{(u_j, v_j)\mid 1 \le j \le i\}$. Let $M^0$ denote the root block of $G'$. 
We will prove that Algorithm \ref{algBiConnect} maintains the following invariants for each $0 \le i \le m$: 
\begin{itemize}
 \item The graph $G'_{i}$ is outerplanar. 
\item When the algorithm has just finished adding the edge $(u_i, v_i)$, the set of touched blocks, ${\bigcup_{1 \le j \le i}\{B^j_0, B^j_1, \ldots, B^j_{k_j}\}}$, have merged together and formed a single block, which we call as the \textbf{merged block} $M^i$ in $G'_{i}$. 
$M^i$ will be taken as the root block of $G'_{i}$. The other blocks of $G'$ remain the same in $G'_i$.
\item If $i< m$, $x^{i+1}_1$ is the successor of $u_{i+1}$ in the Hamiltonian cycle of the block $M^i$. 
\end{itemize}
By Lemma \ref{outerplanarfirst}, $G'_0=G'$ is outerplanar and it is clear that the above invariants hold for $G'_0$. Assume that the invariants hold for each $i$, where $1 \le i < h \le m$. Consider the case when $i=h$. Since the invariants hold for $h-1$, $x^{h}_1$ is the successor of $u_{h}$ in the Hamiltonian cycle of the block $M^{h-1}$. By Property \ref{touchedblocks}, the blocks $\{B^{h}_1, B^{h}_2, \ldots, B^{h}_{k_{h}}\}$ are untouched when the algorithm has just added the edge $(u_{h-1}, v_{h-1})$. Since the invariants hold for $h-1$, these blocks remain the same in $G'_{h-1}$ as in $G'$. Therefore, the path $P^h=(u_h, x^h_1, x^h_2, \ldots, x^h_{k_h}, v_h)$ continues to satisfy the pre-conditions of Lemma \ref{extensionlemma} in $G'_{h-1}$ (Property \ref{prop3}). On addition of the edge $(u_h, v_h)$ to $G'_{h-1}$, the resultant graph $G'_{h}$ is outerplanar, by Lemma \ref{extensionlemma}. 

By Property \ref{touchedblocks}, the blocks $\{B^h_1, B^h_2, \ldots, B^h_{k_h}\}$ are precisely the blocks that were not touched at the time when $e_{h-1}$ was just added but became touched by the time when $e_h$ is just added. However, by Lemma \ref{Hamiltonianunchanged}, the blocks $\{B^h_1, B^h_2, \ldots, B^h_{k_h}\}$ merges with $M^{h-1}$ and forms the block $M^h$ of $G'_h$ and other blocks of $G'_h$ are same as those of $G'_{h-1}$ (and hence of $G'$) when the edge $e_h$ is added. Thus, all touched blocks have merged together to form the block $M^h$ in $G_h$ and the other blocks of $G'$ remain the same in $G_h$. 

Finally, we have to prove that the successor of $u_{h+1}$ in the Hamiltonian cycle of the block $M^h$ is $x^{h+1}_1$, which is the same as the successor $u_{h+1}$ in the Hamiltonian cycle of the block $B^{h+1}_0$ in $G'$. To see this, note that, by Lemma \ref{Hamiltonianunchanged}, if $v'$ is the successor of $v$ in the block containing the edge $(v, v')$ before an edge $(u_j, v_j)$ is added, it remains so after adding this edge, unless $v = u_j$. By Property \ref{initialvertex}, $u_{h+1} \ne u_j$ for any $j<h+1$ and hence it follows that $x^{h+1}_1$ remains the successor $u_{h+1}$ in the block containing the edge $(u_{h+1}, x^{h+1}_1)$ in $G'_h$. Hence, in order to prove that $x^{h+1}_1$ is the successor of $u_{h+1}$ in the Hamiltonian cycle of the block $M^h$, it suffices to prove that the edge $(u_{h+1}, x^{h+1}_1)$ belongs to the block $M^h$ in $G'_h$. In the previous paragraph we saw that, at the time of adding the edge $(u_h, v_h)$, all the touched blocks so far have merged together to form the the 
block $M^h$ of $G'_{h}$. Since the edge $(u_{h+1}, x^{h+1}_1)$ is in the block $B^{h+1}_0$ in $G'$, which is a touched block by Property \ref{blocksconnected} when the algorithm has just finished adding the edge $(u_h, v_h)$, the block $B^{h+1}_0$ has also been merged into $M^h$ and hence, the edge $(u_{h+1}, x^{h+1}_1)$ is in the block $M^h$ in $G'_h$.

Thus, all the invariants hold for $i=h$ and hence for each $1 \le i \le m$. Since $G''=G'_m$ by definition, $G''$ is outerplanar.
\qed
\end{proof}
\begin{lemma}
 $(\mathcal{P}'', \mathscr{X}'')$ is a path decomposition of $G''$ of width at most $16p+15$.
\end{lemma}
\begin{proof}
It is clear that $(\mathcal{P}'', \mathscr{X}'')$ is a path decomposition of $G''$, since we constructed it using the method explained in Section \ref{background}. 

For each $e_i = (u_i, v_i) \in E''\setminus E'$, let $bypassSeq(e_i)=x^i_1, x^i_2, \ldots, x^i_{k_i}$ and let $S_i$ denote the set of cut vertices that belong to $bypassSeq(e_i)$. By Property \ref{pathBetweenPairs}, $u_i, x^i_1, \ldots, x^i_{k_i}, v_i$ is a path in $G'$. 

We will show that, if $t \in Gap_{\mathscr{X}'}(u_i, v_i)$, then, $X'_t \cap S_i \ne \emptyset$. Without loss of generality, assume that $LastIndex_{\mathscr{X}'}(u_i) < FirstIndex_{\mathscr{X}'}(v_i)$. Since $u_i$ is adjacent to $x^i_1$, both of them are together present in some bag $X'_t \in \mathscr{X}'$, with $t \le LastIndex_{\mathscr{X}'}(u_i)$. Similarly, since $v_i$ is adjacent to $x^i_{k_i}$, they both are together present in some bag $X'_t \in \mathscr{X}'$, with $t \ge FirstIndex_{\mathscr{X}'}(v_i)$. Suppose some bag $X'_t \in \mathscr{X}'$ with $t \in Gap_{\mathscr{X}'}(u_i, v_i)$ does not contain any element of $S_i$. Let $U_i = \{ x^i_j \in S_i \mid x^i_j$ belongs to $X'_{t'} \in \mathscr{X}'$ for some $t' < t \}$ and $V_i = \{ x^i_j \in S_i \mid x^i_j$ belongs to $X'_{t'} \in \mathscr{X}'$ for some $t' > t \}$. From the definitions, $x^i_1 \in U_i$ and $x^i_{k_i} \in V_i$. If $U_i \cap V_i \ne \emptyset$, the vertices belonging to $U_i \cap V_i$ will be present in $X'_t$ as well, which is a 
contradiction. Therefore, $(U_i, V_i)$ is a partitioning of $S_i$. Let $q$ be the maximum such that $x^i_q \in U_i$. Clearly, $q < k_i$. Since $(x^i_q, x^i_{q+1})$ is an edge in $G'$, both $x^i_q$ and $x^i_{q+1}$ should be simultaneously present in some bag in $\mathscr{X}'$. But this cannot happen because $x^i_q \in U_i$ and $x^i_{q+1} \in V_i$. This is a contradiction and therefore, if $t \in Gap_{\mathscr{X}'}(u_i, v_i)$, then, $X'_t \cap S_i \ne \emptyset$.

By the modification done to the path decomposition to reflect the addition of an edge $e_i$, a vertex was inserted into $X''_t \in \mathscr{X}''$ only if $t \in Gap_{\mathscr{X}'}(u_i, v_i)$ and for each $X''_t \in \mathscr{X}''$ such that $t \in Gap_{\mathscr{X}'}(u_i, v_i)$, exactly one vertex ($u_i$ or $v_i$) was inserted into $X''_t$ while adding $e_i$. Moreover, when this happens, $X'_t \cap S_i \ne \emptyset$. Therefore, for any $t$ in the index set, $|X''_t| \le |X'_t|+|\{i \mid 1 \le i \le m, S_i \cap X'_t \ne \emptyset\}|$. But, $|\{i \mid 1 \le i \le m, S_i \cap X'_t \ne \emptyset\}| \le |X'_t|$, because $S_i \cap S_j =\emptyset$, for $1 \le i <j \le m$, by Lemma \ref{surjectionLemma}. Therefore, for any $t$, $|X''_t| \le 2 |X'_t| \le 2(8p+8)$. Therefore, width of the path decomposition $(\mathcal{P}'', \mathscr{X}'')$ is at most $16p+15$.
\qed
\end{proof}
\section{Time Complexity}
For our preprocessing, we need to compute a rooted block tree of the given outerplanar graph $G$ and compute the Hamiltonian cycles of each non-trivial block. These can be done in linear time \cite{Chartrand67,Hopcroft73,Syslo79}. The special tree decomposition construction in Govindan et al.\cite{Govindan98} is also doable in linear time. Using the Hamiltonian cycle of each non-trivial block, we do only a linear time modification in Section \ref{stage1}, to produce the nice tree decomposition $(T, \mathscr{Y})$ of $G$ of width $3$. An optimal path decomposition of the tree $T$, can be computed in $O(n \log n)$ time \cite{Skodinis2003}. For computing the nice path decomposition $(\mathcal{P}, \mathscr{X})$ of $G$ in Section \ref{stage1}, the time spent is linear in the size of the path decomposition obtained for $T$, which is $O(n \log{n})$ \cite{Skodinis2003}, and the size of $(\mathcal{P}, \mathscr{X})$ is $O(n \log n)$. Computing the FirstIndex, LastIndex and Range of vertices and the sequence number of 
blocks can be done in time linear in the size of the path decomposition. Since the resultant graph is outerplanar, Algorithm \ref{algMulticut} and Algorithm \ref{algBiConnect} adds only 
a linear number of new edges. Since the size of each bag in the path decompositions $(\mathcal{P}', \mathscr{X}')$ of $G'$ and $(\mathcal{P}'', \mathscr{X}'')$ of $G''$ are only a constant times the size of the corresponding bag in $(\mathcal{P}, \mathscr{X})$, the time taken for modifying $(\mathcal{P}, \mathscr{X})$ to obtain $(\mathcal{P}', \mathscr{X}')$ and later modifying it to $(\mathcal{P}'', \mathscr{X}'')$ takes only time linear in size of $(\mathcal{P}, \mathscr{X})$; i.e., $O(n \log n)$ time. Hence, the time spent in constructing $G''$ and its path decomposition of width $O(\operatorname{pw}(G))$ is $O(n \log n)$. 
\section{Conclusion}
In this paper, we have described a $O(n \log n)$ time algorithm to add edges to a given connected outerplanar graph $G$ of pathwidth $p$ to get a $2$-vertex-connected outerplanar graph $G''$ of pathwidth at most $16p+15$. We also get the corresponding path decomposition of $G''$ in $O(n \log n)$ time. Our technique is to produce a nice path decomposition of $G$ and make use of the properties of this decomposition, while adding the new edges. Biedl \cite{Biedl2012} obtained an algorithm for computing planar straight line drawings of a $2$-vertex-connected outerplanar graph $G$ on a grid of height $O(p)$. 
In conjunction with our algorithm, Biedl's algorithm will work for any outer planar graph $G$. As explained by Biedl \cite{Biedl2012}, this gives a constant factor approximation algorithm to get a planar drawing of $G$ of minimum height. 

\appendix
\section{Paths between consecutive vertices in the Hamiltonian cycle of a $2$-vertex-connected outerplanar graph}
\begin{claim}
Let $G$ be a $2$-vertex-connected outerplanar graph. Then, the number of internally vertex disjoint paths between any two consecutive vertices in the Hamiltonian cycle of $G$ is exactly two. 
\end{claim}
\begin{proof}
Since $G$ is a $2$-vertex-connected outerplanar graph, it can be embedded in the plane, so that its exterior cycle $C$ is the unique Hamiltonian cycle of $G$ \cite{Chartrand67}. Consider such an embedding of $G$ and let $C=$ $(v_1, v_2, \ldots, v_n, v_1)$, where the vertices of the cycle $C$ are given in the clockwise order of the cycle. Consider any pair of of consecutive vertices in $C$. Without loss of generality, let $(v_1, v_2)$ be this pair. The paths $P_1=(v_1, v_2)$ and $P_2=(v_1, v_n, v_{n-1}, \ldots,v_3, v_2)$ are obviously two internally vertex disjoint paths in $G$, between $v_1$ and $v_2$. 

Since the path $P_1=(v_1, v_2)$ is internally vertex disjoint from any other path in $G$ between $v_1$ and $v_2$, it is enough to show that, there cannot be two internally vertex disjoint paths $P_i$ and $P_j$ between $v_1$ and $v_2$ without using the edge $(v_1, v_2)$.
For contradiction, assume that $P_i$ and $P_j$ are two internally vertex disjoint paths between $v_1$ and $v_2$ without using the edge $(v_1, v_2)$. Let $(v_1, v_i)$ be the first edge of $P_i$ and $(v_1, v_j)$ be the first edge of $P_j$. Without loss of generality, assume that $i<j$. This implies that the edge $(v_1, v_i)$ is not an edge of the exterior cycle $C$ and hence, the (curve corresponding to the) edge $(v_1, v_i)$ splits the region bounded by $C$ into two parts. Let the closed region bounded by the path $v_i, v_{i+1}, \ldots, v_n, v_1$ and the edge $(v_1, v_i)$ be denoted by $C_l$ and the closed region bounded by by the path $v_i, v_{i-1}, \ldots, v_1$ and the edge $(v_1, v_i)$ be denoted by $C_r$.

Let the subpath of $P_j$ from $v_j$ to $v_2$ be denoted by $P'_j$. Since $v_j$ is in $C_l \setminus C_r$ and $v_2$ is in $C_r \setminus C_l$, the path $P'_j$ has to cross from $C_l$ to $C_r$ at least once. Since $P_j$ is vertex disjoint from $P_i$, the path $P'_j$ cannot cross from $C_l$ to $C_r$ at $v_i$. Since the path $P_j$ is simple, $P'_j$ cannot cross from $C_l$ to $C_r$ at $v_1$ also. This implies that there is an edge $(u, v)$ in $P'_j$ with $u$ belonging to $C_l \setminus C_r$ and $v$ belonging to $C_r \setminus C_l$. This would mean that the curve corresponding to the edge $(u, v)$ will cross the curve corresponding to the edge $(v_1, v_i)$, which is a contradiction, because by our assumption, our embedding is an outerplanar embedding. Therefore, there cannot be two internally vertex disjoint paths $P_i$ and $P_j$ between $v_1$ and $v_2$ without using the edge $(v_1, v_2)$.

Thus, the number of internally vertex disjoint paths between any two consecutive vertices in the Hamiltonian cycle of $G$ is exactly two.
\qed
\end{proof}
\section{Proof of Lemma \ref{orderEncounter}}
For any cut vertex $y$ of $G'$, let $s(y)$ denote the Hamiltonian cycle successor of $y$ in the (unique) child block at $y$ and let $G'_y$ denote the subgraph of $G'$ induced on the vertices belonging to the child block attached at $y$ and its descendant blocks. We call the variables $v', v, B, completed[$ $], bypass[$ $], completedCount$ the variables relevant for the traversal. First we prove two basic lemmas which makes the proof of Lemma \ref{orderEncounter} easier.
\begin{lemma}\label{easyprop}
At any point of execution, if a non-cut vertex $u$ is encountered, i.e., $v'$ is set to $u$, until a cut vertex is encountered in Line \ref{reassignv'} or $completedCount=|V(G)|$, from each vertex the algorithm proceed to encounter its Hamiltonian successor in the current block, completing it and incrementing the $completedCount$ by one each time.
\end{lemma}
\begin{proof}
If $u$ is encountered in Line \ref{initialVertexLine}, the next vertex is encountered in Line \ref{reassignv'}, inside the outer while-loop. When $v'$ is a non-cut vertex encountered in Line \ref{reassignv'} or Line \ref{childBlockLine}, the inner while-loop condition and the condition in Line \ref{checkLine} will be evaluated to false until a cut vertex is encountered in Line \ref{reassignv'}. Therefore, the variables relevant for the traversal can get updated in lines \ref{completeLine2}-\ref{reassignline} and lines \ref{reassignv'}-\ref{initialSeq} only, until $v'$ gets assigned to refer to a cut vertex. From this, the property follows.
\qed
\end{proof}
\begin{lemma}\label{recursiveTraversal}
Suppose at a certain time $T_1$ of execution, Algorithm \ref{algBiConnect} has just executed Line \ref{innerWhile} and the variable $v'$ is referring to a cut vertex $x$ in $G'$ and the following conditions are also true: 
\begin{itemize}
\item[$C_1$.] The current block being traversed, i.e. the block referred to by the variable $B$, is the parent block at $x$.
\item[$C_2$.] $bypass(x)=$FALSE and for each cut vertex $y$ of $G'_x$, $bypass(y)=$FALSE. 
\item[$C_3$.] For each vertex $y$ of $G'_x$, $completed(y)=$ FALSE and $completedCount=c$, where $c \le |V(G')| - |V(G'_x)|$
\end{itemize}
Then, the algorithm will again come to Line \ref{innerWhile} with the variable $v'$ referring to the same cut vertex $x$. Let $T_2$ be the next time after $T_1$ when this happens. At time $T_2$, the following conditions will be true:
\begin{itemize}
\item[$E_1$.]\label{e1} The current block being traversed is the child block at $x$ and during the time between $T_1$ and $T_2$ the algorithm never sets the variable $B$ to a block other than the child block at $x$ or its descendant blocks. 
\item[$E_2$.]\label{e2} $bypass(x)=$TRUE and for each cut vertex $y$ of $G'_x$, $bypass(y)=$TRUE. 
\item[$E_3$.]\label{e3} $completed(x)=$FALSE and for each vertex $y$ of $G'_x$ other than $x$, $completed(y)=$TRUE and $completedCount=c+|V(G'_x)|-1 < |V(G')|$ .
\item[$E_4$.]\label{e4} The order in which the algorithm encounters vertices during the period from the time $x$ was encountered just before $T_1$ and till the time $T_2$ is $Order(G'_x, x), x$.
\end{itemize}
\end{lemma}
\begin{proof}
We give a detailed proof of this lemma below, which is in principle just a description of the execution of the algorithm. Instead of going through the proof, the reader may verify the correctness of the lemma directly from the algorithm.  

We prove this lemma using an induction on $n(x)$, the number of blocks in $G'_x$. For the base case, assume that $n_x=1$; i.e., $G'_x$ is a leaf block of $G'$. Suppose the assumptions in the statement of the lemma hold at time $T_1$. By this assumption, the condition of the inner while-loop in Line \ref{innerWhile} has been evaluated to true at time $T_1$ and after executing Line \ref{bypassline} $bypass(x)$ will be set to TRUE. Similarly, it follows from the assumptions that after executing Line \ref{childBlockLine} the current block is set as the child block at $x$ and the algorithm sets $v'=s(x)$, the successor of $x$ in the child block at $x$. Since the child block at $x$ is a leaf block, $v'=s(x)$ is not a cut vertex. By Lemma \ref{easyprop}, until $v'$ gets assigned to refer to a cut vertex, from each vertex the algorithm proceed to encounter its Hamiltonian successor in the current block, completing it and incrementing the $completedCount$ by one each time. Note that $completedCount<|V(G')|$ all this 
time, because at time $T_1$, we had $c \le |V(G')| - |V(G'_x)|$ and the number of times $completedCount$ was incremented since time $T_1$ is less than $|V(G'_x)|$. Since the only cut vertex in the current block is $x$ itself, this goes on until $x$ is encountered in Line \ref{reassignv'} and then it reaches Line \ref{innerWhile} at time $T_2$. From this, it follows that the order in which vertices were encountered during the period from the time $x$ was encountered just before $T_1$ and till the time $T_2$ is the Hamiltonian cycle order of the child block at $x$, starting and ending at $x$. This is precisely $Order(G'_x, x), x$. It is evident that conditions $E_1$ - $E_3$ are also true at time $T_2$.

Now, we will assume that the lemma holds for all cut vertices $y$ such that $n(y) < n(x)$. In order to prove the lemma, it is enough to prove that the lemma holds for $x$ as well. Let $x, v_1, v_2, \ldots, v_t, x$ be the Hamiltonian cycle of the child block at $x$. Suppose the assumptions in the statement of the lemma hold at time $T_1$ when the Algorithm \ref{algBiConnect} has just executed Line \ref{innerWhile} and $v'=x$. As in the base case, $bypass(x)$ will be set to TRUE and in Line \ref{childBlockLine} the current block is set as the child block at $x$ and the algorithm sets $v'=s(x)=v_1$. 

If the set $\{v_{1}, v_{2}, \ldots, v_t\}$ does not contain any cut vertices, we are in the base case and we are done. Otherwise, let $l$ be the minimum index in $\{1, 2, \ldots, t\}$ such that $v_l$ is a cut vertex. By Lemma \ref{easyprop}, from each vertex the algorithm proceed to encounter its Hamiltonian successor in the current block, completing it and incrementing the $completedCount$ each time until $v'=v_l$. Notice that $completedCount=c+l-1 \le |V(G')| - |V(G'_x)|+ l-1 \le |V(G')| - |V(G'_{v_l})|$, when $v_l$ is encountered in Line \ref{reassignv'}. After this, the algorithm reaches line \ref{innerWhile} and executes it with $v'=v_l$ at time $T'_l$. At this time, the block being traversed is the parent block at $v_l$. Since $v_l$ or any other vertex in $G'_{v_l}$ were not encountered till now after $T_1$, and by the assumptions of the lemma about the state of the traversal related variables at time $T_1$, we know that at time $T'_l$, $bypass(v_l)=$FALSE and for each cut vertex $z$ of $G'_{v_l}$, 
$bypass(z)=$FALSE. Similarly, for each vertex $y$ of $G'_{v_l}$, $completed(y)=$ FALSE. Thus, the pre-conditions of the lemma are satisfied for the vertex $v_l$ and $G'_l$ at time $T'_l$. 

The order in which the vertices are encountered during the period from the time $x$ was encountered just before $T_1$ and till the time $T'_l$ is $x, v_1, v_2, \ldots, v_l$. Since $v_l$ is a cut vertex in the child block at $x$, $n(v_l) < n(x)$. Therefore, by induction hypothesis, the algorithm will again come to Line \ref{innerWhile} with the variable $v'=v_l$ and if $T''_l$ is the next time this happens after $T_l$, the conditions $E_1$ - $E_4$ will be satisfied with $v_l$ replacing $x$, $T'_l$ and $T''_l$ replacing $T_1$ and $T_2$ respectively and $c+l-1$ replacing $c$. 

At time $T''_l$, when the algorithm is back at Line \ref{innerWhile} and executes the line with $v'=v_l$ and $bypass(v_l)=$TRUE, the while-loop condition will evaluate to FALSE and so, the loop will not be entered. When the algorithm reaches Line \ref{checkLine}, the condition will evaluate to true and therefore, in Line \ref{parentBlockLine}, the variable $B$ will be updated to its parent block. Since $B$ is the child block at $v_l$ before this, $B$ will be updated to the parent block at $v_l$, which is the same as the child block at $x$. In Line \ref{completeLine2}, $completed(v_l)$ is set to TRUE and $completedCount$ becomes $c+l-1+|V(G'_{v_l})| < |V(G')|$ and therefore, in Line \ref{outerWhile}, the outer while-loop condition evaluates to TRUE.  
Since $v'=v_l$ now, the algorithm executes Line \ref{reassignv'}, and $v'$ will be updated to $v_{l+1}$, the successor of $v_l$ in $B$. From the time $x$ was encountered just before $T_1$, the order in which the algorithm has encountered vertices is $x, v_1, v_2, \ldots, v_{l-1}, Order(G'_{v_l}, v_l), v_l, v_{l+1}$. 

By repeating similar arguments as above, we can reach the following conclusion. If $i$ is the maximum index in $\{1, 2, \ldots, t\}$ such that $v_i$ is a cut vertex, the algorithm will come to Line \ref{innerWhile} with the variable $v'=v_i$ at time $T''_i$ such that the conditions below will be true at time $T''_i$. 
\begin{itemize}
\item The current block being traversed is the child block at $v_i$. During the time between $T_1$ and $T''_i$ the algorithm never sets the variable $B$ to a block other than the child block at $x$ or its descendant blocks. 
\item For each cut vertex $y$ of $G'_x$, $bypass(y)=$TRUE. 
\item For $v_k \in \{v_i, v_{i+1}, \ldots, v_t, x\}$, $completed(v_k)=$FALSE and for each vertex $y$ of $G'_x$ outside this set,\\$completed(y)=$TRUE. Moreover, $completedCount=c+i-1+\sum_{1 \le j \le i, v_j\text{ is a cut vertex}}{(|V(G'_{v_j})|-1)}$
\item The order in which the algorithm encounters vertices during the period from the time $x$ was encountered just before $T_1$ and till the time $T''_i$
is given by $x, S_1, \ldots, S_{v_i}$, where for $1 \le j \le i$, $S_j = v_j$ if $v_j$ is not a cut vertex in $G'$ and $S_j=Order(G'_{v_j}, v_j), v_j$ otherwise. 
\end{itemize}
By similar arguments as at time $T''_l$, we can show that, at the time $T''_i$, the inner while-loop condition is false, because $bypass(v_i)=$TRUE and in Line \ref{parentBlockLine} the variable $B$ will be updated to the child block at $x$. In Line \ref{completeLine2}, $completed(v_i)$ is set to TRUE and $completedCount$ becomes $c+i+\sum_{1 \le j \le i, v_j\text{ is a cut vertex}}{(|V(G'_{v_j})|-1)}$. Since this value is less than $|V(G')|$, in Line \ref{outerWhile} the outer while-loop condition evaluates to TRUE. Since $v'=v_i$ now, the algorithm executes Line \ref{reassignv'}, and $v'$ will be updated to the successor of $v_i$ in $B$. 

If $v_i=v_t$, then at this stage, $v'=x$ and when the algorithm reaches Line \ref{innerWhile}, that is the time $T_2$ mentioned in the lemma and the order in which the algorithm has encountered vertices is $x, S_1, \ldots, S_{v_t}, x$, where for $1 \le j \le t$, $S_j = v_j$ if $v_j$ is not a cut vertex in $G'$ and $S_j=Order(G'_{v_j}, v_j), v_j$ otherwise. If $v_i \ne v_t$, by the maximality of $i$ and using Lemma \ref{easyprop}, until $v'$ gets assigned the value $x$ in Line \ref{reassignv'}, from each vertex the algorithm proceed to encounter its Hamiltonian successor in the current block, completing it and incrementing the $completedCount$ each time. We have $completedCount<|V(G')|$ all this time, because at time $T_1$ we had $completedCount=c$ and afterwards $completedCount$ was incremented once for each vertex $y$ in $G'_x$ for which $completed(y)=TRUE$ but $completed(x)=$FALSE still. When $x$ is encountered in Line \ref{reassignv'} and then the algorithm reaches Line \ref{innerWhile}, that is the time 
$T_2$ mentioned in the lemma. From the time when $x$ was encountered just before $T_1$, the order in which the algorithm has encountered vertices is $x, S_1, \ldots, S_{v_t}, x$, where for $1 \le j \le i$, $S_j = v_j$ if $v_j$ is not a cut vertex in $G'$ and $S_j=Order(G'_{v_j}, v_j), v_j$ otherwise. Thus, in both cases, at time $T_2$ the order in which the algorithm has encountered vertices from the time when $x$ was encountered just before $T_1$ is given by $x, S_1, \ldots, S_{v_t}, x = Order(G'_x, x), x$. Notice also that the conditions $E_1$-$E_3$ also hold at time $T_2$ and hence the lemma is proved.
\qed
\end{proof}
\textbf {Lemma \ref{orderEncounter}.}
If $G'$ is given as the input graph to Algorithm \ref{algBiConnect}, where $G'$ is a connected outerplanar graph with at least two vertices, such that for every cut vertex $x$ of $G'$, $G'\setminus x$ has exactly two connected components, and if $v_0$ is the non-cut vertex in the root block of $G'$ from which the algorithm starts the traversal, then $Order(G', v_0)$ is the order in which Algorithm \ref{algBiConnect} encounters the vertices of $G'$. 
\begin{proof}
Suppose $v_0, v_1, \ldots, v_t, v_0$ is the Hamiltonian cycle of the root block of $G'$. Our method is to use Lemma \ref{recursiveTraversal} at each cut vertex in the root block of $G'$ and Lemma \ref{easyprop} at each non-cut vertex in the root block of $G'$, until $completedCount=|V(G')|$. We will see that this will go on until $v_t$ is completed.  

After the initializations done in lines \ref{forLoop} - \ref{completeLine1}, for
every cut vertex $y$ in $G'$, $bypass(y)=$FALSE, current block $B$ is the root block of $G'$,  $v'=v_0$, $completedCount=1$, $completed(v_0)$=TRUE and for every other vertex $y$ in $G'$, $completed(y)$=FALSE. Since the outer while-loop condition is TRUE, the algorithm enters the loop and in Line \ref{reassignv'}, $v_1$ is encountered. Let us call this instant as time $T_1$. 

If the set $\{v_1, v_2, \ldots, v_t\}$ contains a cut vertex, let $l$ be the minimum index in $\{1, 2, \ldots, t\}$ such that $v_l$ is a cut vertex. Otherwise, let $v_l =v_t$. By Lemma \ref{easyprop}, until $v'$ gets assigned the value $v_l$ in Line \ref{reassignv'}, from each vertex the algorithm proceed to encounter its Hamiltonian successor in the current block (i.e the root block), completing it and incrementing the $completedCount$ each time, making $completedCount= l <|V(G')|$ when $v_l$ is encountered in Line \ref{reassignv'}. The algorithm then reaches line \ref{innerWhile} and executes it with $v'=v_l$. Let us call this instant as time $T'_l$. The order in which the vertices are encountered from the beginning of execution of the algorithm is $v_0, v_1, v_2, \ldots, v_l$.

If $v_l=v_t$ and $v_l$ is not a cut vertex, that means the graph $G'$ does not have a cut vertex and $|V(G')|=t+1$. In this case, the inner while-loop condition evaluates to FALSE. Similarly, the condition in Line \ref{checkLine} also evaluates to FALSE. In Line \ref{completeLine2}, the algorithm sets $completed(v_t)=$TRUE and increments $completedCount$, making $completedCount=t+1=|V(G')|$. When the algorithm executes Line \ref{outerWhile}, the outer while-loop condition evaluates to FALSE and the algorithm terminates. The order in which the vertices were encountered from the beginning of execution of the algorithm is $v_0, v_1, \ldots, v_t = Order(G', v_0)$. 

The other case is when $v_l$ is a cut vertex at time $T'_l$. At this time, $completedCount=l \le |V(G')| - |V(G'_l)|$ and by similar arguments as in the proof of Lemma \ref{recursiveTraversal}, the pre-conditions of the lemma are satisfied for the vertex $v_l$ and $G'_l$ at time $T'_l$.
Therefore, by Lemma \ref{recursiveTraversal} applied to the vertex $v_l$ and $G'_l$, the algorithm will again come to Line \ref{innerWhile} with the variable $v'=v_l$. If $T''_l$ is the next time this happens after $T'_l$, we can see that the state of the traversal related variables at time $T''_l$ is similar to those we obtained in the proof of Lemma \ref{recursiveTraversal}, except that $completedCount=l+|V(G'_{v_l})|-1$. 

Following similar arguments as in the proof of Lemma \ref{recursiveTraversal}, we can show that when the algorithm executes Line \ref{parentBlockLine} the next time after $T''_l$, the variable $B$ will be updated to the parent block at $v_l$, which is the same as the root block. In Line \ref{completeLine2}, $completed(v_l)$ is set to TRUE and $completedCount$ becomes $l+|V(G'_{v_l})|$. 

By repeating similar arguments as in the proof of Lemma \ref{recursiveTraversal}, we can reach the following conclusion. If $i$ is the maximum index in $\{1, 2, \ldots, t\}$ such that $v_i$ is a cut vertex, the algorithm will come to Line \ref{innerWhile} with the variable $v'=v_i$ at time $T''_i$ such that the conditions below will be true at time $T''_i$. 
\begin{itemize}
\item The current block being traversed is the child block at $v_i$. 
\item For each cut vertex $y$ of $G'$, $bypass(y)=$TRUE. 
\item For $v_k \in \{v_i, v_{i+1}, \ldots, v_t\}$, $completed(v_k)=$FALSE and for each vertex $y$ of $G'$ outside this set,\\$completed(y)=$TRUE. Moreover, $completedCount=i+\sum_{1 \le j \le i, v_j\text{ is a cut vertex}}{(|V(G'_{v_j})|-1)}$
\item The order in which the algorithm encounters vertices from the beginning of execution of the algorithm till the time $T''_i$ is given by $v_0, S_1, \ldots, S_{v_i}$, where for $1 \le j \le i$, $S_j = v_j$ if $v_j$ is not a cut vertex in $G'$ and $S_j=Order(G'_{v_j}, v_j), v_j$ otherwise. 
\end{itemize}
By similar arguments as earlier, when the algorithm executes Line \ref{parentBlockLine} the next time after $T''_i$, the variable $B$ will be updated to the parent block at $v_i$, which is the same as the root block. In Line \ref{completeLine2}, $completed(v_i)$ is set to TRUE and $completedCount$ becomes $i+1+\sum_{1 \le j \le i, v_j\text{ is a cut vertex}}{(|V(G'_{v_j})|-1)}$. 

If the above sum is equal to $|V(G')|$, that means $v_i=v_t$. When the algorithm executes Line \ref{outerWhile}, the outer while-loop condition evaluates to FALSE and the algorithm terminates. The order in which the vertices were encountered from the beginning of execution of the algorithm is $v_0, S_1, \ldots, S_{v_t}$, where for $1 \le j \le t$, $S_j = v_j$ if $v_j$ is not a cut vertex in $G'$ and $S_j=Order(G'_{v_j}, v_j), v_j$ otherwise. This is the same as $Order(G', v_0)$. 

Instead, if the sum is less than $|V(G')|$, then $v_i \ne v_t$. In Line \ref{outerWhile} the outer while-loop condition evaluates to TRUE.  
Since $v'=v_i$ and $B$ is the root block, when the algorithm executes Line \ref{reassignv'}, $v'$ will be updated to $v_{i+1}$, the successor of $v_i$ in the root block. By the maximality of $i$ and using similar arguments as earlier, we can show that until $v'$ gets assigned the value $v_t$ in Line \ref{reassignv'}, from each vertex the algorithm proceed to encounter its Hamiltonian successor in the current block. When $v_t$ is encountered in Line \ref{reassignv'} and then reach line \ref{innerWhile}, the condition of the inner while-loop will evaluate to FALSE because $v_t$ is not a cut vertex. Similarly, the condition in Line \ref{checkLine} will also evaluate to FALSE. In Line \ref{completeLine2}, $completed(v_t)$ is set to TRUE and $completedCount$ becomes $t+1+\sum_{1 \le j \le i, v_j\text{ is a cut vertex}}{(|V(G'_{v_j})|-1)}$, which is equal to $|V(G')|$. When the algorithm executes Line \ref{outerWhile}, the outer while-loop condition evaluates to FALSE and the algorithm terminates. From the 
beginning of execution, the order in which the algorithm has encountered vertices is $v_0, S_1, \ldots, S_{v_t}$, where for $1 \le j \le i$, $S_j = v_j$ if $v_j$ is not a cut vertex in $G'$ and $S_j=Order(G'_{v_j}, v_j), v_j$ otherwise. This order is the same as $Order(G', v_0)$.

In all cases, from the beginning of execution of Algorithm \ref{algBiConnect} till it terminates, the order in which the algorithm encounters the vertices of $G'$ is given by $Order(G', v_0)$.
\qed
\end{proof}
\end{document}